%% file: techreport.tex
\pdfoutput=1
\documentclass[envcountsame]{llncs}

\usepackage{amsmath,amstext,amssymb}

\usepackage{enumerate}

\usepackage{palatino}

\include{preamble}

\input{macros}

\title{A Note on the Complexity of Model-Checking
  Bounded Multi-Pushdown Systems\thanks{Work
    partially supported by projects \emph{ARCUS IdF/Inde} 
   and \emph{EU Seventh Framework Programme} under grant 
   agreement No. PIOF-GA-2011-301166 (DATAVERIF).}}

\author{Kshitij Bansal\inst{1} \and St\'ephane Demri\inst{1,2}}

\institute{New York University, USA \and LSV, CNRS, France}

\newif \ifLONG \LONGtrue
\newif \ifSHORT \SHORTfalse

\newcommand{\compressableeqn}[1]{\ifLONG \[#1\]\fi\ifSHORT $#1$\fi}

\begin{document}

\maketitle
\begin{center}December 5, 2012\end{center}

\input{abstract}

\input{section-introduction}

\input{section-preliminaries}
\input{section-synchronisation}
\input{subsection-repreach}

\input{section-others}

\input{section-conclusion}

\bibliographystyle{abbrv}
\bibliography{biblio}

\end{document}

%% file: preamble.tex
\usepackage{marginnote}
\newcommand{\draftnote}[1]{\marginpar{\footnotesize{#1}}}

\newcommand{\inparen}[1]{\left(#1\right)}
\newcommand{\insqrbr}[1]{\left[#1\right]}
\newcommand{\inbrace}[1]{\left\{#1\right\}}
\newcommand{\abs}    [1]{\left|#1\right|}
\newcommand{\pair}[2]{\inparen{#1,#2}}

\newcounter{savedenum}
\newcommand*{\saveenum}{\setcounter{savedenum}{\theenumi}}
\newcommand*{\resume}{\setcounter{enumi}{\thesavedenum}}

\usepackage{graphicx}

\usepackage{hyperref}

\makeatletter                           
\providecommand*{\toclevel@title}{0}
\providecommand*{\toclevel@author}{0}
\makeatother                            

\hypersetup{
    colorlinks=true,
    linkcolor=blue,
    citecolor=blue,
    filecolor=blue,
    urlcolor=blue,
}

%% file: macros.tex
\newcommand{\multicaret}{Multi-\caret}

\DeclareMathOperator{\pre}{pre}
\DeclareMathOperator{\post}{post}

\newcommand{\mps}{P}            
\newcommand{\mpsStates}{G}      
\newcommand{\astate}{g}
\newcommand{\mpsAlphabet}{\Gamma}
\newcommand{\aletter}{a}
\newcommand{\aletterbis}{b}
\newcommand{\aword}{\mpsWord} 
\newcommand{\mpsTrans}{\Delta}
\newcommand{\mpsStackCount}{N}

\newcommand{\actCall}{\mathsf{call}}
\newcommand{\actRet}{\mathsf{return}}
\newcommand{\actInt}{\mathsf{internal}}

\newcommand{\actions}{\mathfrak{A}}
\newcommand{\anaction}{\mathfrak{a}}

\newcommand{\arun}{\rho}

\newcommand{\aconfig}{c}

\newcommand{\mpsStack}{\mpsAlphabet^*}
\newcommand{\mpsMoves}{\insqrbr{\mpsStackCount}}
\newcommand{\step}[1]{\xrightarrow{}_{#1}}
\newcommand{\cstep}[2]{\xrightarrow{\!\!#1\!\!}_{#2}} 

\newcommand{\mpsState}{g}
\newcommand{\mpsMove}{s}
\newcommand{\mpsWord}{w}
\newcommand{\mpsWordAtom}{v}
\newcommand{\mpsWordReturn}{u}
\newcommand{\mpsChar}{\aletter} 

\newcommand{\aformula}{\phi}
\newcommand{\aformulabis}{\psi}
\newcommand{\ltlState}{\mpsState}   

\newcommand{\next}{\mathsf{X}}
\newcommand{\until}{\mathsf{U}}
\newcommand{\atom}{A}
\newcommand{\opAbs}{\mathsf{a}}          
\newcommand{\opCall}{\mathsf{c}}        

\DeclareMathOperator{\Succ}{succ}
\DeclareMathOperator{\Cl}{Cl}
\DeclareMathOperator{\Atoms}{Atoms}

\newcommand{\vect}[1]{\boldsymbol{#1}}
\newcommand{\prodify}[1]{\widehat{#1}}

\newcommand{\prodmps}{\prodify{\mps}}
\newcommand{\aprodrun}{\prodify{\arun}}

\newcommand{\tNoRet}{\mathsf{noreturn}}
\newcommand{\tRet}{\mathsf{willreturn}}
\newcommand{\tAlive}{\mathsf{alive}}
\newcommand{\tDead}{\mathsf{dead}}
\newcommand{\tagsRet}{\inbrace{\tNoRet,\tRet}}
\newcommand{\tagsDead}{\inbrace{\tAlive,\tDead}}

\newcommand{\artag}{r}
\newcommand{\adtag}{d}

\newcommand{\accSet}{\mathcal{F}}

\newcommand{\set}[1]{\{ #1 \}}

\newcommand{\couple}[2]{( #1,#2 )}
\newcommand{\triple}[3]{( #1,#2,#3 )}

\newcommand{\Nat}{\mathbb{N}}

\newcommand{\bottom}{\perp}
\newcommand{\powerset}[1]{{\cal P}(#1)}

\newcommand{\card}[1]{{\rm card}(#1)}

\newcommand{\aset}{X}
\newcommand{\asetbis}{Y}

\newcommand{\length}[1]{\left|#1\right|}
\newcommand{\size}[1]{\left|#1\right|}

\newcommand{\np}{\textsc{NP}}

\newcommand{\exptime}{\textsc{Exp\-Time}}
\newcommand{\etime}{\textsc{E\-Time}}

\newcommand{\defstyle}[1]{{\em #1}}
\newcommand{\egdef}{\stackrel{\mbox{\begin{tiny}def\end{tiny}}}{=}} 
\newcommand{\equivdef}{\stackrel{\mbox{\begin{tiny}def\end{tiny}}}{\equivaut}} 
\newcommand{\equivaut}{\;\Leftrightarrow\;}
\newcommand{\interval}[2]{[#1,#2]}
\newcommand{\caret}{{\tt CaRet}}

\newcommand{\poststar}[1]{{\rm post}^{\star}(#1)}
\newcommand{\aautomaton}{\mathcal{A}}
\newcommand{\aautomatonbis}{\mathcal{B}}
\newcommand{\alang}{{\rm L}}

\newcommand{\stackStimeT}[2]{\aword_{#1}^{#2}}

\newcommand{\aalphabet}{\Sigma}

\newcommand{\augmentrun}[1]{\gamma(#1)}

%% file: abstract.tex
\begin{abstract}
In this note, we provide complexity characterizations of model checking multi-pushdown systems.
\ifLONG 
Multi-pushdown systems model recursive concurrent programs 
in which any sequential process has a finite control. 
\fi 
We consider \ifLONG three \fi  standard notions for boundedness:
context boundedness, phase boundedness and stack ordering. The logical formalism is a linear-time
temporal logic extending well-known logic \caret \  but dedicated to multi-pushdown systems in which
abstract  operators 
\ifLONG (related to calls and returns) \fi 
such as those for next-time and until are
parameterized by stacks.
We show that the problem is \exptime-complete for context-bounded runs and unary encoding of the number of context
switches; we also prove that the problem is 2\exptime-complete for phase-bounded runs and unary
 encoding of the number of phase switches. In both cases, the value $k$ is given as an 
\ifLONG
input 
(whence it is not a constant of the model-checking problem),
\fi 
 which makes a substantial difference in the complexity. 
\ifLONG
In certain cases, our results improve previous complexity results. 
\fi 
\end{abstract}

%% file: section-introduction.tex
\section{Introduction}
\label{section-introduction}
 
\paragraph{Multi-pushdown systems.}
Verification problems for pushdown systems, systems with a finite automaton and an
unbounded stack, have been extensively studied and decidability can be obtained as in the case
for finite-state systems. Indeed, many problems of interest like computing
$\pre^{\star}(\aset)$ (set of configurations reaching a regular set $\aset$), $\post^{\star}(\aset)$ 
(set of configurations accessible from a regular set $\aset$), reachability and LTL model checking have been shown to
be decidable~\cite{Bouajjani&Esparza&Maler97,Finkel&Willems&Wolper97,Schwoon02,Walukiewicz01}. 
More precisely, existence of infinite runs with B\"uchi acceptance
condition is decidable for pushdown systems. The proof is based on the fact that
$\pre^{\star}(\aset)$ is  computable and regular for any regular set $\aset$
of configuration ($\post^{\star}(\aset)$ is also regular and  computable)~\cite{Bouajjani&Esparza&Maler97}.
These have also been implemented, for instance
in the model-checker Moped~\cite{Schwoon02}. 
It can be
argued that they are natural models for modeling recursive programs. Two
limitations though of the model are the inability to model programs with
infinite domains (like integers) and modeling concurrency. 
Having an infinite automaton to handle the former limitation leads to 
undecidability~\cite{Minsky67}. An approach to tackle this has been to
abstract infinite-state programs to Boolean programs using, for instance,
predicate abstraction. The model is repeatedly refining,
as needed, like in the SLAM tool~\cite{Ball&Rajamani01}, SATABS~\cite{Clarkeetal05}
etc.
For concurrency, a natural way to extend this model would be to consider
pushdown automata with multiple stacks, 
which has seen significant
interest in the recent past~\cite{Atig10,Atig10bis,BCGZ11,Esparza&Ganty11}.
This is the main
object of study in this paper which we call \defstyle{multi-pushdown systems}.

\paragraph{The difficulty of model-checking multi-pushdown systems.}
A pushdown system with even two stacks and with a singleton stack alphabet is sufficient to model a Turing
machine (see e.g. this classical result~\cite{Minsky67}), hence making the problem of even
testing reachability undecidable.
%
%
This is not a unique situation  and similar issues exists with other
abstractions, like model-checking problems on counter systems;
other models of multithreaded programs
 are also known
to admit undecidable verification problems. 
That is why subclasses of runs have been introduced 
as well as problems related to the search for 'bounded runs' that may satisfy a 
desirable or undesirable property. 
\ifSHORT
\else
For instance,
reversal-bounded counter automata have effectively semilinear reachability sets~\cite{Ibarra78},
context-bounded model-checking (bound on the number of context switches)~\cite{Qadeer&Rehof05}
(see also the more recent work dealing with complexity issues in~\cite{Esparza&Ganty11}),
and of course boun\-ded model-checking (BMC)
(bound on the distance of the reached positions)~\cite{Biereetal03}.
\fi 

\paragraph{Motivations.}
In \cite{Qadeer&Rehof05}, \np-completeness of the reachability problem restricted
to context-bounded runs was shown. This was also implemented in a tool called ZINC
to verify safety properties and find bugs in actual programs in few context switches,
showing that feasibility of the approach. Since, there has been significant work on 
considering weaker restrictions
and other related problems.
This paper focuses on the study of model-checking problems
for multi-pushdown systems based on LTL-like dialects, naturally
allowing to express liveness properties, when some bounds are fixed.
Though decidability of 
these problems has been established in some recent works 
\ifLONG
(as a consequence
of considering more expressive logics like monadic second-order logic), 
\fi 
we aim to provide optimal computational
complexity analysis for LTL-like properties. In particular, we consider a LTL-like specification
language based on \caret~\cite{Alur&Etessami&Madhusudan04}, which strikes to us
as fitting given the interest of the model in program verification.

\paragraph{Content.}
We consider the logic Multi-\caret, an extension of
\caret \  
\ifLONG
(which itself is a generalization of LTL)
\fi  
to be able to reason
about runs of concurrent recursive programs.
Next, we study the model-checking problem
of \multicaret \ formulae over multi-pushdown systems restricted to $k$-bounded
runs and give an \exptime \ upper bound, when $k$ is encoded in unary. Since
this problem is a generalization of LTL model checking pushdown systems which
is known to be \exptime-hard, this is an optimal result. Viewed as an extension 
of~\cite{Bouajjani&Esparza&Maler97},
we consider both a more general model and a more general logic, while still
preserving the complexity bounds.
%
%
At a technical level, we focus on combining several approaches in order to achieve optimal complexity
bounds. In particular, we combine the approach taken in \caret \ model-checking of 
recursive state machines 
\ifLONG
 machines 
(equivalent to pushdown systems in expressiveness, but
more explicitly model recursive programs)~\cite{Alur&Etessami&Madhusudan04},
\fi 
ideas from reachability analysis of multi-pushdown
systems~\cite{Schwoon02} 
and the techniques introduced in work on pushdown systems in~\cite{Bouajjani&Esparza&Maler97,Schwoon02}.
Also, we go further and consider less restrictive notions of boundedness that have
been considered in the literature, and obtain optimal or near-optimal complexity bounds for these.
To summarize the results in the paper,
\begin{itemize}
\itemsep 0 cm 
\item 
      Multi-\caret  \ model-checking over multi-pushdown systems with $k$-\textbf{context bounded runs}~\cite{Qadeer&Rehof05}
      is \exptime-complete  when $k$ is encoded in unary and it is in  2\exptime \  if the encoding is binary.
      The value $k$ is given as an input and not as a parameter of the problem,  
      which makes a substantial difference when complexity analysis is provided.
\item Multi-\caret \  model-checking over  multi-pushdown systems with $k$-\textbf{phase boun\-ded 
      runs}~\cite{LaTorre&Madhusudan&Parlato07} is
      \ifLONG 
      in 2\exptime \ when $k$ encoded in unary, and it is in 3-\exptime \ if the encoding is binary.
      \fi 
      Note that this problem can be encoded from developments in~\cite{BCGZ11} but
      the \exptime \ upper bound from~\cite{BCGZ11} applies when the number of phases
      is fixed. Otherwise, one gets an 3\exptime \ upper bound if $k$ is encoded in binary.
\item Similarly, Multi-\caret \  model-checking over \textbf{ordered
      multi-pushdown systems}~\cite{Atig&Bollig&Habermehl08}
       is
      \ifLONG
      in 2\exptime \ when $k$ encoded in unary, and it is in 3\exptime \ if the encoding is binary.
      \fi 
      This is the best we can hope for 
       since by~\cite{Atig10bis}, reachability problem, 
       existence of infinite runs with B\"uchi acceptance condition and LTL model-checking
       are decidable for ordered multi-pushdown systems, the latter problem being 2\etime-complete. 
       \ifLONG 
       Global model-checking is also decidable~\cite[Theorem 12]{Atig10bis}.
       \fi 
\end{itemize}

In~\cite{Madhusudan&Parlato11}, decidability results can be found for \ifLONG several \fi 
classes of automata with auxiliary storage based on MSO property. This includes multi-pushdown systems with 
bounded context and ordered multi-pushdown systems. This might also include problems
with temporal logics as stated in~\cite{Madhusudan&Parlato11}. 
\ifLONG
Another work that is worth looking at is the one on 
games on multi-stack systems~\cite{Seth09} where parity games on bounded 
      multi-stack systems are shown  decidable thanks to a  new proof for the decidability of
      the emptiness problem.
\fi 
  Though, the complexity is non-elementary in the size of
  the formula, arrived at by using celebrated 
  \ifLONG Courcelle's Theorem~\cite{Courcelle90}, 
  \fi 
 which has
  parameterized complexity non-elementary, the parameter being the size of
  formula plus the tree-width.
  Non-elementary lower bounds can be also reached with branching-time logics, see 
  e.g.~\cite{Atigetal12}. 

\paragraph{Comparison with two recent works.}
In this paper, we generalize the automata-based approach for LTL
to a linear-time temporal logic for multi-pushdown systems. A similar
approach have been also followed in the recent works~\cite{Atigetal12bis,LaTorre&Napoli12}.
Let us explain below the main differences with the present work. 

In~\cite{Atigetal12bis}, LTL model-checking on multi-pushdown systems when runs are $k$-scope-bounded
is shown \exptime-complete. Scope-boundedness strictly extends context-boundedness and therefore
Corollary~\ref{corollary-k-boundedness-reg}(I) and~\cite[Theorem 7]{Atigetal12bis} are closely related
even though  Corollary~\ref{corollary-k-boundedness-reg}(I) deals with Multi-\caret \  and it takes into
account only context-boundedness. 
Moreover, we are able to deal with regular constraints on stack contents
while keeping the optimal complexity upper bound. 
By contrast,~\cite{LaTorre&Napoli12} introduces an extension of \caret \ that is identical
to the variant we consider in our paper. Again,~\cite{LaTorre&Napoli12} deals with scope-boundedness
 and Corollary~\ref{corollary-k-boundedness-reg}(I) and~\cite[Theorem 6]{LaTorre&Napoli12} are closely related
even though Corollary~\ref{corollary-k-boundedness-reg}(I) takes into account only context-boundedness, 
which leads to a slightly different result. 
The upper bound for OBMC from Corollary~\ref{corollary-obmc} is a relatively simple consequence of
the way we can reduce model-checking to repeated reachability with generalized B\"uchi acceptance
conditions. Similarly,~\cite[Theorem 7]{LaTorre&Napoli12} provides an optimal complexity upper bound
for ordered multiply nested words, whence  Corollary~\ref{corollary-obmc} and~\cite[Theorem 7]{LaTorre&Napoli12}
are also  related. 

As a concluding remark, the work presented in this note is partly subsumed by the recent developments
presented in~\cite{Atigetal12bis,LaTorre&Napoli12}. 
Nevertheless, the upper bounds in Corollary~\ref{corollary-k-boundedness}(I) and in 
Corollary~\ref{corollary-k-boundedness-reg}(I) are original results apart from the way we build the synchronized
product in Section~\ref{section-synchronisation}. 
Moreover,  we believe that our developments
shed some useful light on technical issues. For instance, we deal with context-boundedness, phase-boundedness
and ordered multi-push\-down systems uniformly while providing in several cases optimal complexity upper
bounds. Our complexity analysis for context-boundedness relies on~\cite{Bouajjani&Esparza&Maler97,Schwoon02}
whereas for ordered multi-pushdown systems it relies on~\cite{Atig10bis} (for instance, this contrasts
with developments from~\cite[Section 5]{LaTorre&Napoli12}). Finally, our construction
allows us to add regularity constraints, that are known to go beyond first-order language, by a simple 
adaptation of the case for Multi-\caret.

%% file: section-preliminaries.tex
\section{Preliminaries}
\label{section-preliminaries}

We write $\insqrbr{N}$ to denote the set $\set{1, 2, \ldots, N}$.
We also use a boldface  as a shorthand for elements indexed by $\mpsMoves$,
for e.g., $\vec{a}$ $=$ $\set{a_i \ |\ i \in \mpsMoves}$.
For an alphabet $\aalphabet$, $\aalphabet^*$ represents the set of finite words over $\aalphabet$, 
$\aalphabet^+$ the set of finite non-empty words over $\aalphabet$.  
For a finite word $\aword=\aletter_1\ldots \aletter_k$ over $\aalphabet$, 
we write $\length{\aword}$ to denote its \defstyle{length} $k$. 
For $0 \leq i < \length{\aword}$, $\aword(i)$ represents the $(i+1)$-th letter of the word, here $\aletter_{i+1}$. We use 
$\card{\aset}$ to denote the number of elements of a finite set $\aset$.

\subsection{Multi-Pushdown Systems}
\label{section-multipushdownsystems}

\ifLONG
In this section,  we first define multi-pushdown systems 
and then present a simple reduction into multi-pushdown systems with global states in 
 which the next active stack can be read. There exists a  correspondence
in terms of traces (which is what we need for the forthcoming model-checking
problems).
\fi 

Pushdown systems provide a natural execution model for programs with
recursion. A generalization with multiple stacks allows us to model threads.
The formal model is described below, which we will call
\defstyle{multi-pushdown systems}.
\begin{definition}
A \defstyle{multi-pushdown system}  is a tuple of the form 
$$\mps = 
 (\mpsStates, \mpsStackCount, \mpsAlphabet, \mpsTrans_1, 
 \ldots ,\mpsTrans_{\mpsStackCount}),$$ for some $N \geq 1$ such that:
\ifLONG
\begin{itemize}
\itemsep 0 cm
\item $\mpsStates$ is a non-empty finite set of  \defstyle{global states},
\item $\mpsAlphabet$ is the finite \defstyle{stack alphabet} containing the distinguished letter $\bottom$,
\item for every $s \in \insqrbr{\mpsStackCount}$, $\mpsTrans_s$ is the \defstyle{transition relation}
      acting on the $s$-th stack where $\mpsTrans_s$  
      is a relation included in $\mpsStates \times \mpsAlphabet \times \mpsStates \times \actions(\mpsAlphabet)$
      with $\actions(\mpsAlphabet)$ defined as 
$$\actions(\mpsAlphabet)  \egdef \bigcup_{\aletter \in \mpsAlphabet}\set{\actCall(\aletter),
\actRet(\aletter), \actInt(\aletter)}$$
      Elements of the set $\actions(\mpsAlphabet)$ are to be thought of as \defstyle{actions} modifying the stack with alphabet $\mpsAlphabet$. 
\end{itemize}
\fi 
\end{definition}
A \defstyle{configuration} $\aconfig$ of the multi-pushdown system $\mps$ 
is the global state   along with contents of the $\mpsStackCount$ stacks, i.e. $\aconfig$ belongs to
$\mpsStates \times (\mpsStack)^\mpsStackCount$. 
For every $s \in  \mpsMoves$, we write $\step{s}$ to denote the \defstyle{one-step relation} with respect to
the $s$-th stack. Given two configurations 
 $\aconfig =(\mpsState, \mpsWord_1, \ldots, \mpsWord_s \mpsChar, \ldots
\mpsWord_\mpsStackCount)$ and $\aconfig'=(\mpsState', \mpsWord_1, \ldots,
\mpsWord_s', \ldots, \mpsWord_\mpsStackCount)$, $\aconfig \step{s} \aconfig'$ 
$\equivdef$ $(\mpsState, \mpsChar, \mpsState', \anaction(\aletterbis)) \in \mpsTrans_s$
where $\anaction(\aletterbis)$ reflects the change in the stack enforcing one of the conditions below:
\ifLONG
\begin{itemize}
\itemsep 0 cm 
\item $\aword_{s} = \aword_{s}'$, $\anaction = \actRet$ and $\aletter=\aletterbis$,
\item $\aword_{s}'  = \aword_s \aletterbis$ and $\anaction = \actInt$,
\item $\aword_{s}' = \aword_s \aletter \aletterbis$ and $\anaction = \actCall$.
\end{itemize}
\fi 
The letter $\bottom$ from the stack alphabet plays a special role; indeed the initial content
of each stack is precisely $\bottom$. Moreover, $\bottom$ cannot be pushed, popped or replaced
by any other symbol. This is a standard way to constrain the transition relations and to check for `emptiness' of
the stack (i.e. equal to $\bottom$). 
We write $\step{\mps}$ to denote the relation $(\bigcup_{s \in \mpsMoves} \step{s})$.  
Note that given a configuration $\aconfig$, there may exist $\aconfig_1$, $\aconfig_2$ and $i_1 \neq i_2 \in 
\insqrbr{\mpsStackCount}$ such that $\aconfig \step{i_1} \aconfig_1$ and  $\aconfig \step{i_2} \aconfig_2$,
which is the fundamental property to consider such models as adequate for modeling concurrency. 
An infinite \defstyle{run} is an $\omega$-sequence of configurations 
$\aconfig_0, \aconfig_1, \aconfig_2, \ldots$ such that for every $i \geq 0$, we have 
$\aconfig_i \step{\mps} \aconfig_{i+1}$. 
If $\aconfig_i \step{s} \aconfig_{i+1}$, then we say that for that 
step, the $s$-th stack is \defstyle{active}. Similar notions can be defined for finite  runs.
As usual, we write $\aconfig \cstep{*}{} \aconfig'$ whenever there is a finite run from
$\aconfig$ to $\aconfig'$. 
\ifLONG
A standard problem on multi-pushdown systems is the state reachability problem
defined below:
\begin{description}
\itemsep 0 cm
\item[input:] $(\mps, \aconfig, \astate)$ where $\mps$ is a multi-pushdown system, $\aconfig$ is a configuration of $\mps$ 
and $\astate$ a global state of $\mps$.
\item[question:] is there a finite run from $\aconfig$ to some configuration $\aconfig'$ such that
the global state of $\aconfig'$ is $\astate$?
\end{description}
\fi 

An \defstyle{enhanced} multi-pushdown system is a multi-pushdown system of the form
$\mps = 
 (\mpsStates \times \mpsMoves, \mpsStackCount, \mpsAlphabet, \mpsTrans_1,\ldots, \mpsTrans_{\mpsStackCount})$ such that
for every $s \in \mpsMoves$, $\mpsTrans_s \subseteq 
 (\mpsStates \times \set{s}) \times \mpsAlphabet \times (\mpsStates  \times \mpsMoves) \times \actions(\mpsAlphabet)$.
In such multi-pushdown systems, the global state contains enough information to determine the 
next \defstyle{active} stack. 
Observe that the way the one-step relation is defined, we do not necessarily need to carry this information
as part of the finite control (see Lemma~\ref{lemma:mpsdefs}). We do
that in order to enable us to assert about active stack in our logic (see Section~\ref{section-caret}),
and for technical convenience.

\begin{lemma}
\label{lemma:mpsdefs}
Given a multi-pushdown system $\mps = (\mpsStates, \mpsStackCount, \mpsAlphabet,
  \vec{\mpsTrans})$, one can construct in polynomial time
an enhanced multi-pushdown system $\mps' = (\mpsStates \times \mpsMoves, \mpsStackCount, \mpsAlphabet,
  \vec{\mpsTrans'})$ such that
\begin{description}
\itemsep 0 cm
\item[(I)] For every infinite run of $\mps$ of the form
\ifLONG 
  \[ \aconfig_0 \step{s_0} \aconfig_1 \step{s_1}
     \cdots \aconfig_t \step{s_t} \aconfig_{t+1} \cdots \]
\fi 
there exists an infinite run $\aconfig_0' \step{s_0} \aconfig_1' \step{s_1}
     \cdots \aconfig_t' \step{s_t} \aconfig_{t+1}' \cdots$ of $\mps'$ 
such that $(\star)$ for $t \geq 0$, if $\aconfig_t =  (\astate_t,\inbrace{\stackStimeT{s}{t}}_s)$, then 
 $\aconfig_t' =  (\pair{\astate_t}{s_t},\inbrace{\stackStimeT{s}{t}}_s)$.
\item[(II)] Similarly, for every infinite run of $\mps'$ of the form
\ifLONG
\[ \aconfig_0' \step{s_0} \aconfig_1' \step{s_1}
     \cdots \aconfig_t' \step{s_t} \aconfig_{t+1}' \cdots \]
\fi 
there exists an infinite run $\aconfig_0 \step{s_0} \aconfig_1 \step{s_1}
     \cdots \aconfig_t \step{s_t} \aconfig_{t+1} \cdots$ of $\mps$
such that  ($\star$). 
\end{description}
\end{lemma} 

The proof is by an easy verification. In the sequel, 
\ifLONG without any loss of generality,
\fi 
 we shall consider 
enhanced multi-pushdown systems only since all the properties that can be expressed in our logical
languages are linear-time properties. 
\ifLONG
For instance, there is a logarithmic-space reduction
from the state reachability problem to its restriction to enhanced multi-pushdown systems.
\fi 

\subsection{Standard Restrictions on Multi-Pushdown Systems}

State reachability problem is known to be undecidable by  
a simple reduction from the non-emptiness problem for intersection
of context-free grammars. This has motivated works on the definition of
restrictions on runs so that decidability can be regained (for state reachability problem but also
for model-checking problems). We recall below three standard notions
for boundedness; other notions can be found in~\cite{LaTorre&Napoli11,Cyriac&Gastin&Kumar12}. 
Definitions are provided for infinite runs
but they can be easily adapted to finite runs too. 

In the notion of $k$-boundedness defined below, a phase is understood as
a sub-run such that a single stack is active (see e.g.~\cite{Qadeer&Rehof05}).
 
\begin{definition}
\label{definition-k-boundedness}
Let $\arun = \aconfig_0 \step{s_0} \aconfig_1 \step{s_1}
     \cdots \aconfig_t \step{s_t} \aconfig_{t+1} \cdots$ be an infinite run and $k \geq 0$.
We say that $\arun$ is \defstyle{$k$-bounded} if 
there exist positions $i_1 \leq i_2 \leq 
  \ldots \leq i_{k-1}$ such that
  $s_{t} = s_{t+1}$ for all $t \in \Nat \setminus \set{i_1 \ldots i_{k-1}}$.
\end{definition}

In the notion of $k$-phase-boundedness defined below, a phase is understood as
a sub-run such that return actions are performed on a single stack, see 
e.g.~\cite{LaTorre&Madhusudan&Parlato07}.

\begin{definition}
\label{definition-k-phase-boundedness}
Let $\arun = \aconfig_0 \step{s_0} \aconfig_1 \step{s_1}
     \cdots \aconfig_t \step{s_t} \aconfig_{t+1} \cdots$ be an infinite run and $k \geq 0$.
We say that $\arun$ is \defstyle{$k$-phase-bounded} if there is a partition $\asetbis_1, \ldots,
\asetbis_{\alpha}$ of $\Nat$ with $\alpha \leq k$ such that
for every $j \in \interval{1}{\alpha}$ there is $s \in  \mpsMoves$ 
\ifLONG such that \fi
for every $i \in \asetbis_{j}$, if a return action is performed from $\aconfig_i$ to
$\aconfig_{i+1}$, then it is done on the $s$th stack.
\end{definition}

In the notion of order-boundedness defined below, the stacks are linearly ordered and a return action
on a stack can only be performed if the smallest stacks are empty,  see 
e.g.~\cite{Atig&Bollig&Habermehl08}.

\begin{definition}
\label{definition-order-boundedness}
Let $\mps$ be a multi-pushdown system and $\preceq = \pair{\mpsMoves}{\leq}$ be a (finite) total ordering of the
stacks.  
Let $\arun = \aconfig_0 \step{s_0} \aconfig_1 \step{s_1}
     \cdots \aconfig_t \step{s_t} \aconfig_{t+1} \cdots$ be an infinite run.
We say that $\arun$ is \defstyle{$\preceq$-bounded} if 
for every $t \in \Nat$ that a return is performed on the $s$-th stack,
all the stacks strictly smaller than $s$ with respect to  $\preceq$
are empty.
\end{definition}


\section{A \ifLONG Rich Specification \fi Language for Multi-Pushdown Systems: \multicaret}
\label{section-caret}

Below, we introduce \multicaret, an extension of the logic \caret \ proposed in~\cite{Alur&Etessami&Madhusudan04}, and
dedicated to runs of multi-pushdown systems (instead of for runs of recursive state machines as done 
in~\cite{Alur&Etessami&Madhusudan04}). 
For instance, the logical language can state that a stack is active. 
Moreover,
the temporal operators are sometimes parameterized by a stack; for instance, the abstract next binary relation
can be naturally extended to the case when several stacks are present. Note that the logic presented below can
be easily seen as a fragment of monadic second-order logic and therefore the decidability results
from~\cite{Madhusudan&Parlato11,BCGZ11} apply to the forthcoming model-checking problems.
However, our definition makes a compromise between a language of linear-time temporal properties
that extends the logic from~\cite{Alur&Etessami&Madhusudan04} and the most expressive logic for which
our model-checking problems are known to be decidable. Indeed, we aim at proposing optimal 
complexity characterizations. 
The logic presented below is identical the one presented in~\cite{LaTorre&Napoli12} except 
for the presence of regular constraints.

\subsection{Definition}
\label{section-definition-logic}

Models of Multi-\caret \ are infinite runs of multi-pushdown systems. For each (enhanced) multi-pushdown system 
$\mps = 
 (\mpsStates\times \mpsMoves, \mpsStackCount, \mpsAlphabet, \mpsTrans_1, 
 \ldots ,\mpsTrans_{\mpsStackCount})$, we define the fragment \multicaret($\mps$) of \caret \ that uses syntactic
resources from $\mps$ (namely $\mpsStates$ and $\mpsMoves$). 
The full language \multicaret~is defined as the union of all the sub-languages \multicaret($\mps$). 
Formulae of \multicaret($\mps$) are defined according to the grammar
\ifLONG
 below:
\begin{eqnarray*}
  \aformula & := & \ltlState    
                    \ |\ s
                    \ |\ \actCall
                    \ |\ \actRet
                    \ |\ \actInt
                    \\ & & 
                    \ |\ \aformula \vee \aformula 
                    \ |\ \neg \aformula
                    \\ & &
                    \ |\ \next \aformula 
                    \ |\ \aformula \until \aformula
                    \ |\ \next_s^\opAbs \aformula 
                    \ |\ \aformula \until_s^\opAbs \aformula
                    \ |\ \next_s^\opCall \aformula
                    \ |\ \aformula \until_s^\opCall \aformula
\end{eqnarray*}
where $s \in \mpsMoves$, $\ltlState \in \mpsStates$.
\fi 
Models of \multicaret($\mps$) formulae
are $\omega$-sequences  in $\inparen{ \mpsStates \times \mpsMoves \times (\mpsStack)^\mpsStackCount }^\omega$, 
which can be obviously understood as infinite runs of $\mps$.
\paragraph{Semantics.}
Given an infinite run $\arun = \aconfig_0 \aconfig_1 \ldots \aconfig_t \ldots$
with $\aconfig_t = (\mpsState_t, \mpsMove_t, \stackStimeT{1}{t},\ldots, \stackStimeT{\mpsStackCount}{t})$ 
for every position $t \in \Nat$, the satisfaction relation $\arun, t \models \aformula$ with $\aformula$ in 
Multi-\caret($\mps$) is defined inductively as follows (successor relations are defined just below):
{\allowdisplaybreaks

\begin{align*}
&  \arun, t \models \ltlState &
    \text{ iff } & 
    \mpsState_t = \mpsState
    \\
&  \arun, t \models \mpsMove &
    \text{ iff } &
    \mpsMove_t = \mpsMove\\
&  \arun, t \models \anaction &
    \text{ iff } &
    \pair{\anaction}{\length{\stackStimeT{\mpsMove_t}{t+1}} -\length{\stackStimeT{\mpsMove_t}{t}}} \in 
    \set{\pair{\actCall}{1}, \pair{\actInt}{0}, \pair{\actRet}{-1}} \\
&  \arun, t \models \aformula_1 \vee \aformula_2 &
    \text{ iff } &
    \arun,t \models \aformula_1 \text{ or } \arun,t \models \aformula_2
    \\
&  \arun,t \models \neg \aformula &
    \text{ iff } &
    \arun, t \not\models \aformula
    \\
&  \arun,t \models \next \aformula &
    \text{ iff } &
    \arun, \Succ_\arun(t) \models \phi
    \\
& \arun, t \models \aformula_1 \until \aformula_2 &
    \text{ iff } &
    \text{there is a sequence of positions } i_0 = t \text{, }i_1 \ldots\text{, }
    i_k \text{, s.t.} \\
&    & & \text{for } j < k\text{, } i_{j+1} = 
    \Succ_\arun(i_j)\text{, } \arun, i_j \models \aformula_1  \text{ and } 
    \arun,i_k \models \aformula_2
    \\
\intertext{For $b \in \set{\opAbs,\opCall} $ and $s \in \mpsMoves$:}
&  \arun,t \models \next^b_s \aformula &
    \text{ iff } & \Succ^{b,s}_\arun(t)  \ \text{is defined and} \ 
    \arun, \Succ^{b,s}_\arun(t) \models \phi
    \\
&  \arun, t \models \aformula_1 \until^a_s \aformula_2 &
    \text{ iff } &
    \text{there exists a sequence of positions } t \leq i_0 < i_1 \\
&    & &\cdots < i_k \text{, where } i_0 \text{ smallest such with } s_{i_0} = s, \text{ for}\\
&    & &  j < k\text{, } i_{j+1} =  \Succ^{a,s}_\arun(i_j)\text{, }
    \arun, i_j \models \aformula_1  \text{ and} \ \arun,i_k \models \aformula_2
    \\
&  \arun, t \models \aformula_1 \until^c_s \aformula_2 &
    \text{ iff } &
    \text{there exists a sequence of positions } t \geq i_0 > i_1 \\
&    & &\cdots > i_k \text{, where } i_0 \text{ greatest such with } s_{i_0} = s, \text{ for}\\
&    & &  j < k\text{, } i_{j+1} =  \Succ^{c,s}_\arun(i_j)\text{, }
    \arun, i_j \models \aformula_1  \text{ and} \ \arun,i_k \models \aformula_2
    \\
\end{align*}
}

The above definition for $\models$ distinguishes three successor relations: 
\defstyle{global} successor relation, the \defstyle{abstract} successor relation that jumps  to the
first future position after a return action at the same level, if any, and the \defstyle{caller} successor
relation that jumps to the latest past position before a call action at the same level, if any. 
Here are the formal definitions:
\begin{itemize}
\itemsep 0 cm 
\item $\Succ_\arun(t) \egdef t+1$ for every $t \in \Nat$ (standard).
\item $\Succ^{\opAbs,s}_\arun(t)$ is defined below whenever $s$ is active at position $t$:
  \begin{enumerate}
  \itemsep 0 cm
  \item If $\length{\stackStimeT{s}{t+1}} = \length{\stackStimeT{s}{t}} + 1$ (call), then 
   $\Succ^{\opAbs,s}_\arun(t)$  is the smallest $t'>t$ such
  that 
  $s_{t'} = s$ and 
  $\length{\stackStimeT{s}{t'}} = \length{\stackStimeT{s}{t}}$. If there is no such $t'$ then 
  $\Succ^{\opAbs,s}_\arun(t)$  is undefined.

  \item If $\abs{\stackStimeT{s}{t+1}} = \abs{\stackStimeT{s}{t}}$ (internal),
  then $\Succ^{\opAbs,s}_\arun(t)$ is the smallest $t'>t$ such that $s_{t'} = s$ (first position when $s$th stack is active).

  \item If $\abs{\stackStimeT{s}{t+1}} = \abs{\stackStimeT{s}{t}} - 1$ (return), 
   then $\Succ^{\opAbs,s}_\arun(t)$  is undefined.
  \end{enumerate}
\item $\Succ^{\opCall,s}_\arun(t)$ (caller of $s$-th stack): largest $t'<t$ such
  that $s_{t'} = s$ and $\abs{\stackStimeT{s}{t'}} = \abs{\stackStimeT{s}{t}} - 1$. If such a $t'$ does not
  exist, then $\Succ^{\opCall,s}_\arun(t)$  is undefined.
\end{itemize}

\noindent
In the sequel, we write $\arun \models \aformula$ whenever $\arun, 0 \models \aformula$. 

\paragraph{Adding regularity constraints.} Regularity constraints are the most 
natural and simple constraints on stack contents and still such properties
are not always expressible in first-order 
logic (or equivalently in plain LTL).
\ifLONG
Such constraints have a second-order flavour
thanks to the close relationship between MSO and regular languages.
We define Multi-\caret$^{\sc reg}$ as 
the extension
of Multi-\caret \ in which regularity constraints on stack contents can be expressed. 
\fi 
Logic Multi-\caret$^{\sc reg}$  
is defined 
from Multi-\caret \ by adding atomic formulae of the form ${\tt in}(s, \aautomaton)$ where 
$s$ is a stack identifier
and $\aautomaton$ is a \ifLONG finite-state automaton \fi  
over the stack alphabet $\mpsAlphabet$. The \ifLONG satisfaction \fi relation $\models$
is extended accordingly: $\arun,t \models \ {\tt in}(s, \aautomaton)$ $\equivdef$ 
$\stackStimeT{s}{t} \in \alang(\aautomaton)$ where $\alang(\aautomaton)$ is the set of finite words
accepted by $\aautomaton$. 
Note that regularity constraints can be expressed on
each stack. 
Even though most of the developments in the paper are done with Multi-\caret, 
we shall see that all our complexity
upper bounds still hold true with Multi-\caret$^{\sc reg}$.
This is despite the fact that these new constraints have a second-order flavour.

\paragraph{Another set of temporal operators.} 
Temporal operators $\next^a_s$ and $\until_s^a$ in Multi-\caret \ not only are abstract operators
that refer to future positions reached after returns but also they are parameterized by stacks. 
We made the choice to present these operators for their expressive power but also because
they are quite handy in forthcoming technical developments. Below, we briefly present 
the alternative operators $\next_s$,
$\next^a$ and $\until^a$ and we show  how they are related to the operators from Multi-\caret. 

\begin{itemize}
\itemsep 0 cm 
\item  $\arun,t \models \next_s \aformula$ $\equivdef$ there is $t' > t$ such that $s_{t'} = s$ 
       and for the smallest $t'$, we have $\arun, t' \models \aformula$. 
      So, $\next_s \aformula$  states that the next position
      when the stack $s$ is active (if any), $\aformula$ holds true. 
\item $\arun,t \models \next^a \aformula$ $\equivdef$
$\Succ^{a,s_t}_\arun(t)$ is defined and $\arun, \Succ^{a,s_t}_\arun(t) \models \aformula$.
So, $\next^a \aformula$ states that next time the current stack performs a return action, 
 $\aformula$ holds true. 
\item $\arun, t \models \aformula_1 \until^a \aformula_2$  $\equivdef$
    there exists a sequence of positions $t = i_0 < i_1 \cdots < i_k$ where for
    $j < k$, $i_{j+1} =  \Succ^{a,s_t}_\arun(i_j)$,
    $\arun, i_j \models \aformula_1$ and $\arun,i_k \models \aformula_2$. 
\end{itemize}

\ifLONG
Let us write $s$ to denote the formula (disjunction of atomic formulae) stating that the current
active stack is $s$. Note that ($\equiv$ denotes logical equivalence):
$$
\next_s \aformula \equiv (\neg s \until (s \wedge \aformula)) \ \ \
\next^a \aformula \equiv (\bigwedge_{s} (s \Rightarrow \next_s^a \aformula)) \ \ \
\aformula_1 \until^a \aformula_2 \equiv  (\bigwedge_{s} (s \Rightarrow \aformula_1 \until_s^a \aformula_2))
$$

Similarly, we have the following equivalences:
$$
\next^a_s \aformula \equiv (\neg s \Rightarrow \next_s \aformula) \wedge
                           (s \Rightarrow \next^a \aformula) \ \ \
\aformula_1 \until_s^a \aformula_2 \equiv
(\neg s \Rightarrow \next_s (\aformula_1 \until^a \aformula_2)) \wedge
(s \Rightarrow \aformula_1 \until^a \aformula_2) \ \ \
$$

Hence, it is worth noting that the choice we made about the set of primitive operators, does not strictly decrease
the expressive power but we shall see that the operators from Multi-\caret \ happen to be  extremely helpful
in forthcoming technicalities. 
\fi

\subsection{Decision Problems}

Let us introduce the model-checking problems considered in the paper.
\ifLONG
\noindent
Model-checking problem for multi-pushdown systems (MC):
\begin{description}
\itemsep 0 cm
\item[input:] $(\mps, \astate_0, \aformula)$ where $\mps$ is a multi-pushdown system $\mps$,
$\astate_0$ gives an initial configuration $\pair{\astate_0}{(\bottom)^N}$,
$\aformula$ is a formula in Multi-\caret($\mps$).
\item[question:] Is there an infinite run $\arun$ from 
$\pair{\astate}{(\bottom)^N}$ such that $\arun \models \aformula$?
\end{description}
\fi
We know that \ifLONG the model-checking problem for multi-pushdown systems \fi 
is undecidable
whereas its restriction to a single stack is \exptime-complete~\cite{Alur&Etessami&Madhusudan04}.
Now, let us turn to bounded model-checking problems. 
\ifLONG
\noindent
Bounded model-checking problem for multi-pushdown systems (BMC):
\begin{description}
\itemsep 0 cm
\item[input:] $(\mps, \astate_0, \aformula, k)$ where $\mps$ is a multi-pushdown system $\mps$,
$\astate_0$ gives an initial configuration $\pair{\astate_0}{(\bottom)^N}$,
$\aformula$ is a formula in Multi-\caret($\mps$) and $k \in \Nat$ is a natural number thought of as a bound.
\item[question:] Is there an infinite $k$-bounded run $\arun$ from 
$\pair{\astate}{(\bottom)^N}$ such that $\arun \models \aformula$?
\end{description}
\fi 
Note that $k \in \Nat$ is an input of the problem and not a parameter of BMC.
This makes a significant difference for complexity  since usually complexity can increase
when passing from being a constant to being an input. 

Phase-bounded model-checking problem (PBMC) is defined similarly by replacing in the above definition
'$k$-bounded run' by '$k$-phase-bounded run'. Similarly, we can obtain a definition with order-boundedness. 
\ifLONG
\noindent
Order-bounded model-checking problem for multi-pushdown systems (OBMC):
\begin{description}
\itemsep 0 cm
\item[input:] $(\mps, \astate_0, \aformula, \preceq)$ where $\mps$ is a multi-pushdown system $\mps$,
$\astate_0$ gives an initial configuration $\pair{\astate_0}{(\bottom)^N}$,
$\aformula$ is a formula in Multi-\caret($\mps$) and $\preceq = \pair{\mpsMoves}{\leq}$ is a total ordering of the
stacks.
\item[question:] Is there an infinite $\preceq$-bounded run $\arun$ from 
$\pair{\astate}{(\bottom)^N}$ such that $\arun \models \aformula$?
\end{description}
\fi 

We present below the problem of repeated reachability of multi-pushdown systems, denoted REP.
In Section~\ref{section-synchronisation},
we present how MC can be reduced to REP while obtaining optimal complexity upper bounds.
\noindent
\begin{description}
\itemsep 0 cm
\item[input:]  $(\mps, I_0, \mathcal{F})$
where $\mps$ is a multi-pushdown system,
$I_0$ is a subset of global states of $\mps$ denoting the initial states, and
$\mathcal{F}$ is a collection of B\"uchi acceptance sets,
\item[question:] Is there an infinite run $\arun$ from some
  $\pair{\astate_0}{(\bottom)^N}$ with $\astate_0 \in I_0$ such
  that for each $F \in \mathcal{F}$ there exists a $\mpsState_f \in F$
  that is repeated infinitely often?
\end{description}
We will refer to problem restricted to $k$-bounded runs by BREP.
Obviously, the variants with other notions of boundedness can be defined
too.

Finally, the simplified version of Multi-\caret \ consists of the restriction of Multi-\caret \ 
in which atomic formulae are of the form $\pair{g}{s}$ when enhanced multi-pushdown
systems are involved. Logarithmic-space reductions exist between the full problems 
and their restrictions to the simplified languages.

 
\begin{lemma} For every 
\ifLONG problem 
\fi 
$\mathcal{P}$ in $\{$ MC, BMC, PBMC, OBMC $\}$, there is a logarithmic-space
reduction to $\mathcal{P}$ restricted to formulae from the simplified language.
\end{lemma}
The proof is by an easy verification and its very idea consists in adding to global states
information about the next active stack and about the type of action.
In the sequel, without any loss of generality, we restrict ourselves to the simplified languages. 

\begin{theorem} \cite{Madhusudan&Parlato11}
BMC, PBMC and OBMC are decidable.
\end{theorem} 

Decidability proof from~\cite{Madhusudan&Parlato11} is very general and partly relies on Courcelle's Theorem.
However, it provides non-elementary complexity upper bounds. 
As a main result of the paper, we shall show that BMC is \exptime-complete when $k$ is encoded in unary
and in 2\exptime \ when $k$ is encoded in binary. 

%% file: section-synchronisation.tex
\section{From Model-Checking to Repeated Reachability}
\label{section-synchronisation}

Herein, we reduce the problem of model checking (MC) to the
problem of repeated reachability (REP) while noting complexity
features that are helpful later on (Theorem~\ref{theorem:mc2rep}).
This generalizes the reduction from LTL model-checking 
\ifLONG for finite-state systems
\fi 
into non-emptiness for generalized B\"uchi
automata (see e.g.~\cite{Vardi&Wolper94}), similarly to the approach followed in~\cite{LaTorre&Napoli12}; 
not only we have to tailor
the reduction to Multi-\caret \ and to multi-pushdown systems but also
we aim at getting tight complexity bounds afterwards.
The instance of the problem MC that we have is a multi-pushdown system
$\mps$, a formula $\aformula$ and initial state $\pair{g_0}{i_0}$. For the
instance of REP we will reduce to, we will denote the multi-pushdown
system by $\prodmps$, the set of acceptance sets by $\mathcal{F}$ and
set of initial states by $I_0$.


\ifLONG
\subsection{Augmented Runs}
\fi 

Let $\arun$ be a run of the multi-pushdown system 
$\mps = (\mpsStates \times \mpsMoves, \mpsStackCount,
\mpsAlphabet, \vec{\mpsTrans})$ with 
$\arun \in (\mpsStates \times \mpsMoves \times (\mpsAlphabet^{*})^{\mpsStackCount})^{\omega}$.
The multi-pushdown system  $\prodmps$ is built in such a way that its runs correspond
exactly to runs from $\mps$ but augmented with pieces of information related to the satisfaction
of subformulae (taken from the closure set $\Cl(\aformula)$ elaborated on shortly),
whether a stack is dead or not (using a tag from $\tagsDead$) and whether
the current call will ever be returned or not (using a tag from $\tagsRet$).
These additional tags will suffice to reduce the existence of a run satisfying  $\aformula$ to  
the existence of a run satisfying a generalized B\"uchi condition. 
First, we define from $\arun$ an ``augmented run'' $\augmentrun{\arun}$ which is an infinite sequence from 
$\inparen{\prodify{\mpsStates} \times \mpsMoves \times (\prodify{\mpsAlphabet}^{*})^{\mpsStackCount}}^{\omega}$
where $\prodify{\mpsStates} = \mpsStates
\times \powerset{\Cl(\aformula)}^\mpsStackCount \times \tagsRet^\mpsStackCount
\times\allowbreak\tagsDead^\mpsStackCount$
and $\prodify{\mpsAlphabet}=\mpsAlphabet \times \powerset{\Cl(\aformula)} \times
\tagsRet$. 
By definition, an augmented run is simply an $\omega$-sequence
but it remains to check that indeed, it will be also a run of the new system. 
 We will see that  $\prodify{\mpsStates} \times \mpsMoves$ is the set of global states of
 $\prodmps$ and $\prodify{\mpsAlphabet}$ is the stack alphabet of $\prodmps$. 

Before defining $\augmentrun{\cdot}$ which maps runs to augmented runs, let us introduce 
the standard notion for \defstyle{closure} but slightly tailored to our needs.
Note that each global state is partially made of sets of formulas that
can be viewed as obligations for the future. \ifLONG In order to consider only
runs satisfying a formula $\aformula$, it is sufficient to require
that at the first position, obligations include the satisfaction of
$\aformula$. Obligations can be enforced by the transition relation
but also by the satisfaction of B\"uchi acceptance conditions.
It is our intention that the set of runs of $\prodmps$ satisfying such generalized B\"uchi acceptance condition
correspond exactly to the set of augmented runs obtained from runs of $\mps$.
Not only the new system simulates all the runs of the original system but it also keeps
track of which subformulas holds true at each position.
So, the projection of $\augmentrun{\arun}$ over $\mpsStates \times \mpsMoves$ and $\mpsAlphabet$
(i.e, the operation of getting rid of the tags) correspond exactly to $\arun$.

With this understanding of our intentions, we define
obligations as a tuple of sets of formulas indexed by the
stacks; each stack comes with a finite set of formulas. These
formulas are obtained from the closure of $\aformula$, defined as the
set of subformulas of $\aformula$ enriched with formulas for the until
formulas.\fi~This is similar to what is usually done for LTL and is just
a variant of 
\ifLONG Fischer-Ladner closure~\cite{Fischer&Ladner79}.  
\fi 
Given a formula $\aformula$, its \defstyle{closure}, denoted
$\Cl(\aformula)$, is the smallest set that contains $\aformula$,
the elements of $G \times \mpsMoves$, and satisfies the following properties 
($b \in \set{a,c}$ and $s \in \mpsMoves$):
\begin{itemize} 
\itemsep 0 cm 
  \item If $\neg \aformula' \in \Cl(\aformula)$ or $\next \aformula'
    \in \Cl(\aformula)$ or  $\next^b_s \aformula'
    \in \Cl(\aformula)$ then $\aformula' \in \Cl(\aformula)$.
  \item If $\aformula' \vee \aformula'' \in \Cl(\aformula)$, then
    $\aformula', \aformula'' \in \Cl(\aformula)$.
  \item If $\aformula' \until \aformula'' \in \Cl(\aformula)$, then
    $\aformula', \aformula''$, and $\next (\aformula' \until
    \aformula'')$ are in $\Cl(\aformula)$.
   \item If $\aformula' \until_s^b \aformula'' \in \Cl(\aformula)$, then
    $\aformula'$, $\aformula''$, and $\next^b_s (\aformula' \until \aformula'')$
    are in $\Cl(\aformula)$.
  \item If $\aformula' \in \Cl(\aformula)$ and $\aformula'$ in not of the form
    $\neg \aformula''$, then $\neg \aformula' \in \Cl(\aformula)$.
\end{itemize}
Note that the number of formulas in $\Cl(\aformula)$ is linear in the size 
 of $\aformula$ and $\mps$.
 An \defstyle{atom} of $\aformula$, is a set $\atom \subseteq
\Cl(\aformula)$ that satisfies the following
properties:
\begin{itemize}
\itemsep 0 cm 
\item For $\neg \aformula' \in \Cl(\aformula)$,
  $\aformula' \in \atom$ iff
  $\neg \aformula' \notin \atom$.
\item For $\aformula' \vee \aformula'' \in \Cl(\aformula)$,
  $\aformula' \vee \aformula'' \in \atom$ iff
  $(\aformula' \in \atom$ or $\aformula'' \in \atom)$.
\item For $\aformula' \until \aformula'' \in \Cl(\aformula)$,
  $\aformula' \until \aformula'' \in \atom$ iff
  $\aformula'' \in \atom$ or $(\aformula' \in \atom$ and
  $\next (\aformula' \until \aformula'') \in \atom)$.
\item $\atom$ contains exactly one element from $\mpsStates \times \mpsMoves$. 
\end{itemize}
Let $\Atoms(\aformula)$ denote the set of atoms of $\aformula$, along
with empty set (used as special atom, use will become clear
later). Note that there are $2^{\mathcal{O}(\size{\aformula})}$ atoms of
$\aformula$.


We write $\inparen{(\mpsState^t, \mpsMove^t), \vect{\mpsWord}^t}$ to
denote the $t$-th configuration of $\arun$.
We define the augmented run $\augmentrun{\arun}$ so that its $t$-th configuration 
is of the form 
 $\inparen{
      \inparen{\prodify{\mpsState^t}, \mpsMove^t},
      \prodify{\vect{\mpsWord}^t} }$
with
  $\prodify{\mpsState^t} =
    (\mpsState^t, \vect{\atom}^t, \vect{\artag}^t, \vect{\adtag}^t)$
and
 $\prodify{\mpsWord^t_j} =
     (\mpsWord^t_j,\mpsWordAtom^t_j, \mpsWordReturn^t_j)$
for every $j$ in $\mpsMoves$.
We say that the stack $j$ is \emph{active} at time $t$ if $\mpsMove^t = j$.
\newcommand{\deadAliveDescription}
{Then, we define $\tDead$-$\tAlive$ tag to be
  $\tDead$ if and only if the stack is not active at or after the corresponding position.}
\newcommand{\deadAliveEquation}
{& & \forall t\geq 0\text{, }j\in\mpsMoves\text{: }(\adtag^t_j =\tDead) \ \equivdef \ (\forall t' \geq t, \  \mpsMove^{t'} \neq j)\text{.}
\label{sync:cond:dead}}

\newcommand{\atomDescription}{
The idea of the closure as we discussed is to maintain the set of subformulas
that hold true at each step. We
will expect it to be the empty set if the stack is dead.}

\newcommand{\atomEquation}
{\nonumber
& &
 \forall t \geq 0\text{, }j\in\mpsMoves\text{ with }\adtag^t_j = \tAlive\text{, }\aformulabis \in \Cl(\aformula)\text{:}\\
& &
 \qquad\aformulabis \in \atom^t_j \equivdef \arun, t' \models \aformulabis\text{ where }t'\text{ is the least }t' \geq t  \text{ with } \mpsMove^{t'} = j\text{.}
\label{sync:cond:atom}\\
& &
 \forall t \geq 0\text{, }j\in\mpsMoves\text{ with }\adtag^t_j = \tDead\text{: } \atom^t_j \egdef \emptyset\text{.}
\label{sync:cond:atom:b}}

\newcommand{\returnNoreturnDescription}
{As for $\tRet$-$\tNoRet$ tag, it reflects whether a $\actCall$ action has a ``matching'' return.
This is similar to the $\{\infty,ret\}$ tags in~\cite{Alur&Etessami&Madhusudan04}.
This may be done by defining tag to be $\tNoRet$ if stack will never become
smaller than what it is now.}

\newcommand{\returnNoreturnEquation}
{
\nonumber
& & \forall t\geq 0\text{, }j \in\mpsMoves
  \text{ with }\adtag^t_j = \tAlive\text{: }\\
& &
\qquad
(\artag^t_j = \tNoRet) \equivdef 
(\forall t' \geq t, \  \length{w^{t'}_j} \geq \length{w^t_j})\text{.}
\label{sync:cond:return} }

\newcommand{\stackContentDescription}
{Finally, the formulas and $\tRet$-$\tNoRet$ tag on the stack
are defined to be what they were in the global state at the time when the
corresponding letter was pushed on the stack.}

\newcommand{\stackContentEquation}
{\nonumber
& &\forall t \geq 0\text{, }j \in \mpsMoves\text{: }
  \mpsWordAtom^t_j \egdef \atom^{t_1}_j\atom^{t_2}_j\ldots\atom^{t_l}_j\text{
   and  }\mpsWordReturn^t_j \egdef \adtag^{t_1}_j\adtag^{t_2}_j\ldots\adtag^{t_l}_j
\text{,}
\label{sync:cond:stackatom} \\
& & \qquad \text{where for $k$ in $[l]$: $t_k$ is largest $t_k \leq t$
such that $\length{w^{t_k}_j} = k-1$.} }

\ifLONG
\deadAliveDescription \begin{eqnarray}\deadAliveEquation \end{eqnarray}

\atomDescription \begin{eqnarray}\atomEquation\end{eqnarray}

\returnNoreturnDescription \begin{eqnarray}\returnNoreturnEquation\end{eqnarray}

\stackContentDescription \begin{eqnarray}\stackContentEquation\end{eqnarray}
\fi

\ifLONG
We observe below properties about the way tags are placed in $\augmentrun{\arun}$. Later on, 
we shall establish that these conditions are sufficient to guarantee that any sequence satisfying 
these conditions correspond to  a run of $\mps$ whose augmented run is exactly the sequence.
Here are properties of $\augmentrun{\arun}$ that are easy to check using the definition of $\augmentrun{\cdot}$ and
the satisfaction relation $\models$. 

\input{movable-constructedrun-semantics}
\fi

\subsection{Synchronized Product}
\label{subsection-synchronization}
Let us define the multi-pushdown system $\prodify{\mps}$ as $(\prodify{\mpsStates} \times \mpsMoves, \mpsStackCount,
\prodify{\mpsAlphabet}, \vec{\prodify{\mpsTrans}})$ with
\ifLONG
\begin{itemize}
\itemsep 0 cm
\item  $\prodify{\mpsStates} = \mpsStates \times \Atoms(\aformula)^\mpsStackCount \times \tagsRet^\mpsStackCount
\times\allowbreak\tagsDead^\mpsStackCount$, 
\item $\prodify{\mpsAlphabet}=\mpsAlphabet \times \Atoms(\aformula) \times
\tagsRet$,
\item
\fi
 each transition relation $\prodify{\mpsTrans_s}$ is defined such that
$(\prodify{\mpsState}, \mpsMove, \prodify{\mpsChar},
  \prodify{\mpsState'}, \mpsMove', \anaction(\mpsChar'))$
is in $\prodify{\mpsTrans_s}$ $\equivdef$ the conditions  from Figure~\ref{figure:newTrans}.
      are satisfied.

\ifLONG
    These conditions are actually the syntactic counterparts of the semantical properties stated a bit earlier.  
    To refer to elements in the set
    we use $\mpsState$ to denote a state in original multi-pushdown system, $\atom_i$ for atoms,
    $\artag_i$ for return/no-return tags, $\adtag_i$ for dead/alive tags,
    $\mpsChar_S$ to denote stack letter from original multi-pushdown system,
    $\mpsChar_A$ for stack atom, $\mpsChar_\artag$ for return tag saved on the stack.
    We use unprimed version to denote state and letter on top of stack before the transition is
    taken and primed ones for after. Finally, $\anaction$ denotes the action to perform on the stack,
    one of $\{\actCall,\actInt,\actRet\}$.
\end{itemize}
\fi

\begin{figure}
\label{figure:newTrans}
\input{sync-figure-trans-reln}
\caption{Conditions for the transition relation  $\prodify{\mpsTrans_s}$.
We recall that
$\next \atom = \set{ \aformulabis\ |\ \next\aformulabis \in \atom }$,
$\next^\opAbs_1\atom = \set{\aformulabis\ |\ \next^\opAbs_1 \aformulabis \in \atom}$.
}
\end{figure}

The set $\mathcal{F}$ is defined by the following sets of accepting states:
\begin{enumerate}[(a)]
\itemsep 0 cm 
\item \label{buchi:until} 
      For each until formula $\aformulabis = \aformula_1 \until \aformula_2 \in \Cl(\aformula)$, we define 
      \compressableeqn{F^1_{\aformulabis} \egdef \set{(\prodify{\mpsState},\mpsMove) \ |
     \ \aformula_2 \in \atom_\mpsMove \allowbreak\text{ or }\allowbreak
     \aformulabis \notin \atom_\mpsMove}.}
\item \label{buchi:absuntil} 
      For each abstract-until formula $\aformulabis= \aformula_1 \until^\opAbs_\mpsMove \aformula_2 \in \Cl(\aformula)$, we define
       \compressableeqn{F^2_{ \aformulabis} \egdef \set{(\prodify{\mpsState},\mpsMove) \ |
  \  r_\mpsMove = \tNoRet \allowbreak \text{ and } \allowbreak (\aformula_2 \in \atom_\mpsMove
  \text{ or } \aformulabis \notin \atom_\mpsMove) }.}
\item \label{buchi:dead} 
      For each $j \in \mpsMoves$, we define 
      \compressableeqn{ F^3_j \egdef \set{(\prodify{\mpsState},\mpsMove) \ |\ j=\mpsMove } 
       \cup \set{(\prodify{\mpsState},\mpsMove) \ |\ \adtag_j = \tDead}.}
\item \label{buchi:return} 
       For each $j \in \mpsMoves$, \ifLONG we define  \fi 
       \compressableeqn{F^4_j \egdef \set{(\prodify{\mpsState},\mpsMove) \ |\ \adtag_j = \tDead} \cup
            \set{(\prodify{\mpsState},\mpsMove) \ |\ j=\mpsMove, \adtag_\mpsMove = \tNoRet}.}
\end{enumerate}

\ifLONG

Transition relations in $\widehat{\mps}$ and the acceptance conditions in
$\mathcal{F}$ mimic syntactically the semantical properties satisfied
by the augmented runs defined from run of $\mps$. That is why, the
correctness lemma stated below follows from the observation about
$\gamma(\cdot)$ earlier.
\fi
\begin{lemma}
\label{lemma:soundness}
Let $\arun$ be a run of $\mps$. 
Then, $\augmentrun{\arun}$ is a run of $\widehat{\mps}$ such that for every
$F \in \mathcal{F}$, there is a global state in $F$ that is repeated infinitely often.
\end{lemma}

It remains to show that any accepting run corresponds to a run of the
original system, we show that this run in fact the run obtained by
``forgetting'' the augmentations.
\begin{lemma}
\label{lemma:completeness}
Let  $\aprodrun$ be a run of $\prodmps$ satisfying the acceptance condition
$\mathcal{F}$. 
\ifLONG
Then, $\aprodrun$ is the augmented run corresponding
  to $\Pi(\aprodrun)$, which can be shown to be a run of $\mps$:
\[ \gamma(\Pi(\aprodrun)) = \aprodrun \]
\fi 
\end{lemma}

From Lemmas~\ref{lemma:soundness} and~\ref{lemma:completeness} the
soundness and completeness of the reduction follows if we define
the set of new initial states $I_0$ for the REP problem
as states with initial state $\pair{g_0}{i_0}$ for the MC problem and
$\aformula$ present in the part tracking formulae that hold true:
\[ I_0 = \set{ ((\mpsState_0, \vec{\atom}, \vec{\adtag}, \vec{\artag}), i_0) \ |\ 
\aformula \in \atom_{i_0}} \]

This gives an exponential-time reduction
from MC to REP as well as with their bounded variants. 

\begin{theorem}  \label{theorem:mc2rep}
Let $\mps$ be a multi-pushdown system with initial configuration  $(g, (\bottom)^N)$ and $\aformula$
be a Multi-\caret  \ formula. Let $\prodmps$ be the 
\ifLONG multi-pushdown \fi system built from $\mps$, $g$ and $\aformula$,
$I_0$  be the associated set of initial states and $\accSet$ be the acceptance condition. 
\begin{description}
\itemsep 0 cm
\item[(I)] If $\arun_1$ is a run of $\mps$ from  $(g, (\bottom)^N)$ then 
$\arun_2 = \augmentrun{\arun_1}$
is a run of $\prodmps$ satisfying $\accSet$ and (A)-(C) hold true. 
\item[(II)] If $\arun_2$ is a run of $\prodmps$ from some 
\ifLONG initial \fi configuration with global state in  $I_0$ and satisfying
$\accSet$, then $\Pi(\arun_2)$ is a run of $\mps$ and (A)-(C) hold true too.
\end{description}
Conditions (A)--(C) are defined as follows:
\begin{description}
\itemsep 0 cm
\item[(A)] $\arun_1$ is $k$-bounded iff $\arun_2$ is $k$-bounded, for all $k \geq 0$;
\item[(B)] $\arun_1$ is $k$-phase-bounded iff $\arun_2$ is $k$-phase-bounded, for all $k \geq 0$;
\item[(C)] $\arun_1$ is $\preceq$-bounded iff $\arun_2$ is $\preceq$-bounded,  for all
 total orderings of the stacks $\preceq = \pair{\mpsMoves}{\leq}$.
\end{description}
\end{theorem}
Note that at each position, $\arun_1$ and $\arun_2$ work on the same stack and perform the same type of
action (call, return, internal move), possibly with slightly different letters. This is sufficient to guarantee
the satisfaction of the conditions (A)--(C).

\ifLONG In the rest of this section we prove Lemma~\ref{lemma:completeness}.
 \input{proof-completeness-lemma}\fi

%% file: movable-constructedrun-semantics.tex
\begin{enumerate}
\itemsep 0 cm 
\item For all $t \geq 0$ and $j \in \mpsMoves$,  $\adtag^t_j =\tAlive$ if $j = s^t$.
\item 
For all $t \geq 0$ and $j \in \mpsMoves \setminus \set{s^t}$,
$\adtag^{t+1}_j = \adtag^{t}_j$ , $r^{t+1}_j = r^{t}_j$ and
$\atom^{t+1}_j = \atom^{t}_j$ (only tags related to the active stack
may change).
\item 
For all $t \geq 0$ and $j \in \mpsMoves$, $\set{\aformulabis \in \Cl(\aformula): \arun, t \models \aformulabis}$ is an atom
      and therefore for all $t \geq 0$ and $j \in \mpsMoves$, if $\adtag^t_j =\tAlive$ then  $\atom^t_j$ is an atom. 
\item For all $t \geq 0$ and $j \in \mpsMoves$, $(g^t,s^t) \in \atom^t_{s^t}$. 
\item  For all $t \geq 0$, if the $(t+1)$-th action on the stack is a call and
       $r^{t}_{s^t} = \tRet$, then $r^{t+1}_{s^t} = \tRet$ and $\tRet$ is the top symbol
       of $w_{s^t}^{t+1}$. Moreover, $A_{s^t}^{t}$ is the top symbol of $v_{s^t}^{t+1}$.
\item For all $t \geq 0$, if the $(t+1)$-th action on the stack is a return then 
       $r^{t}_{s^t} = \tRet$ and $r^{t+1}_{s^t}$ is equal to the top symbol of $w_{s^t}^{t}$.
      Moreover, $\next_{s^t}^a  \atom \subseteq A_{s^t}^{t+1}$ where $\atom$ is the top symbol of
      $v_{s^t}^{t}$. As a notational convenience, we denote the next formulas in a set $\atom$, using $\next \atom$.
       In other words, $\next \atom = \set{ \aformulabis\ |\ \next\aformulabis \in \atom }$,
       $\next^\opAbs_1\atom \egdef \set{\aformulabis\ |\ \next^\opAbs_1 \aformulabis \in \atom}$ etc.
\item  For all $t \geq 0$, $\next A_{s^t}^{t} \subseteq A_{s^t}^{t+1}$.
\item  For all $t \geq 0$, if the $(t+1)$-th action on the stack is internal 
       then $r^{t+1}_{s^t} = r^{t}_{s^t}$. Moreover,  $\next_{s^t}^a A_{s^t}^{t} \subseteq A_{s^t}^{t+1}$
       and $\next_{s^t}^c A_{s^t}^{t} = \next_{s^t}^c A_{s^t}^{t+1}$.
\item For all $t \geq 0$, the set $\next_{s^t}^a A_{s^t}^t$ is empty if one of the conditions below is met:
      \begin{enumerate}
      \item the $(t+1)$-th action on the stack is a call and $r^{t}_{s^t} = \tNoRet$,
      \item the $(t+1)$-th action on the stack is a return,
      \item  $d^{t+1}_{s^t} = \tDead$. 
      \end{enumerate}
\item For all $t \geq 0$, if the $(t+1)$-th action on the stack is a call, then for every
      $\next_{s^{t}}^c \aformulabis \in \Cl(\aformula)$, 
      $\next_{s^{t}}^c \aformulabis \in A_{s^{t+1}}^{t+1}$ iff $\aformulabis \in A_{s^t}^t$.
\item For all $t \geq 0$ and $j \in \mpsMoves$, if $\aformulabis 
      \in \Cl(\aformula)$ with $\aformulabis = \aformula_1 \until^\opAbs_i \aformula_2$, then:
\begin{enumerate}
\itemsep 0 cm 
\item if $i = j$, then $\aformulabis \in \atom^t_j$ iff $\aformula_1 \in
  \atom^t_j$ or ($\aformula_2 \in \atom^t_j$ and $\next^\opAbs_j \aformulabis \in \atom^t_j$),
\item if $i \neq j$, then $\aformulabis \in \atom^t_j$ iff $\aformulabis \in \atom^t_i$.
\end{enumerate}
\item For all $t \geq 0$ and $j \in \mpsMoves$, if  $\adtag^t_j = \tDead$ then $\atom^t_j = \emptyset$.
\end{enumerate}
Properties stated above are local since they involve at most two successive configurations. It is also
possible to observe B\"uchi conditions that involve the infinite part of $\augmentrun{\arun}$. 

\begin{enumerate}
\itemsep 0 cm 
\item As standard until formulas in LTL, for every  $\aformula_1 \until \aformula_2 \in \Cl(\aformula)$,
      infinitely often in  $\augmentrun{\arun}$ there is a global state  $\inparen{
      \inparen{\prodify{\mpsState^t}, \mpsMove^t},
      \prodify{\vect{\mpsWord}^t} }$
with
  $\prodify{\mpsState^t} =
    (\mpsState^t, \vect{\atom}^t, \vect{\artag}^t, \vect{\adtag}^t)$ such that 
  either $\aformula_2 \in A_{s^t}^t$ or  $\aformula_1 \until \aformula_2  \not \in A_{s^t}^t$. 
\item There is a similar property with abstract until:
      for every $\aformula_1 \until^a_{\mpsMove} \aformula_2 \in \Cl(\aformula)$,
       infinitely often in  $\augmentrun{\arun}$ there is a  state  $\inparen{
      \inparen{\prodify{\mpsState^t}, \mpsMove^t},
      \prodify{\vect{\mpsWord}^t} }$
with
  $\prodify{\mpsState^t} =
    (\mpsState^t, \vect{\atom}^t, \vect{\artag}^t, \vect{\adtag}^t)$ such that 
  $ \mpsMove = \mpsMove^t$, $\artag^t_{ \mpsMove} = \tNoRet$, and 
 $\aformula_2 \in A_{\mpsMove}^t$ or  $\aformula_1 \until^a_{\mpsMove} \aformula_2  \not \in A_{\mpsMove}^t$.
\item  For every $\mpsMove \in \mpsMoves$,  infinitely often in  $\augmentrun{\arun}$ there is a global state  $\inparen{
      \inparen{\prodify{\mpsState^t}, \mpsMove^t},
      \prodify{\vect{\mpsWord}^t} }$
with
  $\prodify{\mpsState^t} =
    (\mpsState^t, \vect{\atom}^t, \vect{\artag}^t, \vect{\adtag}^t)$ such that either $\mpsMove^t = \mpsMove$ or $\adtag^t_{\mpsMove} = \tDead$. 
Moreover, note that $\adtag^t_{\mpsMove} = \tDead$ implies  $\adtag^{t+1}_{\mpsMove} = \tDead$.
\item  For every $\mpsMove \in \mpsMoves$,  infinitely often in  $\augmentrun{\arun}$ there is a global state  $\inparen{
      \inparen{\prodify{\mpsState^t}, \mpsMove^t},
      \prodify{\vect{\mpsWord}^t} }$
with
  $\prodify{\mpsState^t} =
    (\mpsState^t, \vect{\atom}^t, \vect{\artag}^t, \vect{\adtag}^t)$ such that either 
 $\adtag^t_{\mpsMove} = \tDead$ or
($\mpsMove^t = \mpsMove$ and $\adtag^t_{\mpsMove} = \tNoRet$). 
\end{enumerate}

%% file: sync-figure-trans-reln.tex
\begin{multicols}{2}
\begin{enumerate}
\item \label{trans:origTrans}
  $\inparen{(\mpsState,\mpsMove),\mpsChar_S, (\mpsState',\mpsMove'), \anaction(\mpsChar_S') } \in \mpsTrans_\mpsMove$
\saveenum
\end{enumerate}
\begin{enumerate}
\resume
\item \label{trans:activealive}
  $\adtag_\mpsMove = \tAlive$ 
\item \label{trans:nochangeDeadTag}
  $\forall j \neq \mpsMove$, $\adtag_j = \adtag_j'$
\saveenum
\end{enumerate}
\begin{enumerate}
\resume
\item \label{trans:retTagCall}
  If $\anaction = \actCall$, then
  $\artag_\mpsMove = \tRet \Rightarrow \artag_\mpsMove' = \tRet$ and
  $\mpsChar_\artag' = \artag_\mpsMove$
\item \label{trans:retTagInt}
  If $\anaction = \actInt$, then $\artag_\mpsMove' = \artag_\mpsMove$ and
  $\mpsChar_\artag' = \mpsChar_\artag$
\item \label{trans:retTagRet}
  If $\anaction = \actRet$, then $\artag_\mpsMove = \tRet$ and
  $\artag_\mpsMove' = \mpsChar_\artag$
\item \label{trans:nochageRet}
  $\forall j \neq \mpsMove$, $\artag_j = \artag_j'$
\saveenum
\end{enumerate}
\begin{enumerate}
\resume
\item \label{trans:stateConsistent}
  $(\mpsState,\mpsMove) \in \atom_\mpsMove$
\item \label{trans:nochangeAtom}
  $\forall j \neq \mpsMove$, $\atom_j = \atom_j'$
\item \label{trans:globalnext}
  $\next \atom_\mpsMove \subseteq \atom_{\mpsMove'}'$ ($=\atom_{\mpsMove'}$)
\saveenum
\end{enumerate}
\begin{enumerate}
\itemsep 0 cm 
\resume
\item \label{trans:callHandleAbstract}
  If $\anaction = \actCall$, then $\mpsChar_A' = \atom_\mpsMove$
\item \label{trans:internalHandleAbstract}
  If $\anaction = \actInt$, then $\next^\opAbs_\mpsMove \atom_\mpsMove \subseteq \atom_\mpsMove'$ and   $\mpsChar_A' = \mpsChar_A$.
\item \label{trans:returnStackAtom}
  If $\anaction=\actRet$, then $\next^\opAbs_\mpsMove \mpsChar_A \subseteq \atom_\mpsMove'$
\item \label{trans:obligFree}
  Further, $\next^\opAbs_\mpsMove \atom_\mpsMove = \emptyset$ if
\begin{enumerate}
\item \label{trans:noreturnObligFree}
  $\anaction=\actCall$ and $\artag_\mpsMove = \tNoRet$, or
\item \label{trans:popObligFree}
  $\anaction=\actRet$, or
\item \label{trans:dyingObligFree}
  $\adtag_\mpsMove' = \tDead$
\end{enumerate}
\saveenum
\end{enumerate}
\begin{enumerate}
\resume
\item  \label{trans:pushHandleCaller}
  If $\anaction = \actCall$, 
    $\next^\opCall_\mpsMove \atom_{\mpsMove'}' =
    (\next^\opCall_\mpsMove \Atoms(\aformula))
    \cap \atom_\mpsMove$ 
\item \label{trans:internalUnchanged}
  If $\anaction = \actInt$, then
  $\next^\opCall_\mpsMove
  \atom_{\mpsMove'}' = \next^\opCall_\mpsMove \atom_\mpsMove$.
\saveenum
\end{enumerate}
\begin{enumerate}
\resume
\item \label{trans:abstractUntilPropagateActive}
  Let $b \in \{\opAbs, \opCall\}$. Let $\aformulabis \in \Cl(\aformula)$,
  $\aformulabis = \aformula_1 \until^b_s \aformula_2$. Then, $\aformulabis
  \in \atom_s$ iff either $\aformula_2 \in \atom_s$ or ($\aformula_1
  \in \atom_s$ and $\next^\opAbs_s \aformulabis \in \atom_s$).
\item \label{trans:abstractUntilPropagateInactive}
  Let $b \in \{\opAbs, \opCall\}$. Let $\aformulabis \in \Cl(\aformula)$,
  $\aformulabis = \aformula_1 \until^b_j \aformula_2$ with $j \neq s$.
  Then $\aformulabis \in \atom_s$ iff $\aformulabis \in \atom_s'$.
\item \label{trans:deadEmptyAtom}
  $\forall j$: If $\adtag_j = \tDead$, then $A_j = \emptyset$ 
  and $\artag_j = \tNoRet$.
\end{enumerate}
\end{multicols}

%% file: temporary.tex

\begin{proof}[Sketch]
First of all we observe that $\arun \egdef \Pi(\aprodrun)$ is indeed a
run of $\mps$ because of constraint \ref{trans:origTrans} for the
transition relation $\prodify{\mpsTrans}$. Thus, we may conveniently
use the same notation as earlier to denote configurations.
By a case-by-case analysis, we can show that $\aprodrun$ matches each
augmentation of $\arun$ as would be obtained by (1)-(6). The most
interesting case is that of the abstract-until operator we
introduced. We will aim to give a flavor of the techniques and main
idea, a more complete proof may be found in
Appendix~~\ref{section-appendix-completeness-lemma}.
To that end, we assume that we have already established that the stack
contents and the tags are correct (that is those in $\aprodrun$ match
those defined by (\ref{sync:cond:stackatom}) and
(\ref{sync:cond:dead}), (\ref{sync:cond:return}),
(\ref{sync:cond:return:whendead}) respectively).

Consider the case when the atoms differ, if at all possible. Let
$\aformulabis$ denote the smallest formula on which the atoms differ.
Let $t_0$ be the smallest position where the difference is for
$\aformulabis$. Let us consider the case when $\aformulabis =
\aformula_1 \until^\opAbs_i \aformula_2$. Assume $s^{t_0} = i$ (other
cases can be reduced to this case by some analysis). Define $t_1$,
$t_2$, $\ldots$ as sequence of positions with $t_{k+1} =
\Succ^\opAbs_i(t_k)$. It is possible to establish that
$\aformulabis$ is in $\atom^{t_0}_i$, $\atom^{t_1}_i$,
$\atom^{t_2}_i$ etc., and also that $\arun, t_0 \models \aformula_1$,
$\arun, t_1 \models \aformula_1$ etc. This requires arguments
akin to LTL-model checking for pushdown systems, where obligations are

it is possible to establish
\end{proof}

%% file: proof-completeness-lemma.tex
\paragraph{Notation.}
Let $\aprodrun$ be a run of $\prodmps$ satisfying the
acceptance condition $\mathcal{F}$.
First of all, we observe that $\arun \egdef \Pi(\aprodrun)$ is indeed a
run of $\mps$ because of constraint \ref{trans:origTrans} for the
transition relation $\prodify{\mpsTrans}$. Thus, we may conveniently
use the same notation as earlier to denote configurations:
$\inparen{(\mpsState^t, \mpsMove^t), \vect{\mpsWord}^t}$ for $t$-th
configuration of $\arun$,
and
$\inparen{(\prodify{\mpsState^t}, \mpsMove^t), \prodify{\vect{\mpsWord}^t}}$
for $\aprodrun$ with
$\prodify{\mpsState^t} = (\mpsState^t, \vec{\atom}^t, \vec{\adtag}^t, \vec{\artag}^t)$
and  $\prodify{\mpsWord^t_j} =
     (\mpsWord^t_j,\mpsWordAtom^t_j, \mpsWordReturn^t_j)$
for every $j$ in $\mpsMoves$.

\paragraph{Proof strategy.}
By a case-by-case analysis, we will show that $\aprodrun$ matches each
augmentation of $\arun$ as would be obtained by (1)-(6). We will show
this first for contents of the stack, then for the two kinds of
tags. The order will be important, as for some proofs we would need
the assumption that the other parts of the run match. Eventually, we
will show the case for atoms. From the technical aspect, the case for
correctness of parameterized-abstract-until we introduced is the most
interesting, for which we will provide details. The other cases, we shall
just point to the relevant ingredients in the construction, and leave
the details to the reader.

\paragraph{\textbf{Case 1}: Stack content differs.}
Let $t$
denote the first position where augmented stack content in $\aprodrun$
differs from $\augmentrun{\arun}$ as defined in
(\ref{sync:cond:stackatom}).  Let $j$ denote a stack that differs.
First, we observe that the length of the stacks must be the same as
the actions for both runs are the same (this, of course, can be proven
more formally -- but in interest of readability, where clear, we use
informal arguments as these).  Further, since $t$ is the first
position they differ, the difference has to be in the contents at the
``top of the stack''. According to the semantics of multi-pushdown systems, the
change in the stack is possible only in the one stack active at
time $t-1$. We do a case analysis on the action:
\begin{itemize}
\item If $\anaction^{t-1} = \actCall$, then according to constraints
  \ref{trans:retTagCall} and \ref{trans:callHandleAbstract} the
  character at the top of the stack will be the same as defined in
  (\ref{sync:cond:stackatom}), a contradiction.
\item If $\anaction^{t-1} = \actInt$, the stack atom and stack return tag
  will not change from time $t-1$ because of constraints
  \ref{trans:retTagInt} and \ref{trans:internalHandleAbstract},
  contradicting that it differs from $\augmentrun{\arun}$ defined in
  (\ref{sync:cond:stackatom}).
\item If $\anaction^{t-1} = \actRet$, they could not possibly be
  different, since for both the top character is popped.
\end{itemize}

Thus, the augmented stack contents must be
identical and if at all, the difference must be in the \emph{augmented
  global state}. We already observed that the global state and active
stack match, it remains to show that the tags and the atoms match.

\paragraph{\textbf{Case 2}: $\tDead$-$\tAlive$ tag differs.}
Let $t$ denote the least such where the two tags may differ.
Fix $j \in \mpsMoves$, a stack corresponding to which the tag differ at
this position $t$.
\paragraph{\textbf{Case 2(a)}: $\adtag^t_j = \tDead$ but $\exists t' \geq t$ such that $\mpsMove^{t'} = j$.}
We rule this out by the following observation about $\aprodrun$:
\begin{claim}
  Let $j \in \mpsMoves$. Let $\adtag^t_j = \tDead$. Then,
  for all $t' \geq t$,
  $\mpsMove^{t'} \neq j$ and $d^{t'}_j = \tDead$.
\end{claim}
\begin{proof}
  If $d^t_j = \tDead$, then from $\aprodrun$ being a run of $\prodmps$
  constraint \ref{trans:activealive} forces that $\mpsMove^t \neq
  j$. Consequently, from constraint \ref{trans:nochangeDeadTag} we
  conclude that $d^{t+1}_j = d^t_j = \tDead$. By induction on time the
  claim follows.
\end{proof}
\paragraph{\textbf{Case 2(b)}: $\adtag^t_j = \tAlive$ but $\forall t' \geq t$ $\mpsMove^t \neq j$.}
We make another observation which follows directly from
constraint \ref{trans:nochangeDeadTag} that states a tag may not change unless the corresponding stack
is active.
\begin{claim}
  Let $j \in \mpsMoves$. Let $\adtag^t_j = \tAlive$ and further assume $\forall t' \geq t$
  $\mpsMove^{t'} \neq j$.  Then, $\adtag^{t'}_j = \tAlive$ for all $t' \geq t$.
\end{claim}
As a consequence, for the tags to mismatch it must be the case that
for all $t' \geq t$, $\adtag^{t'}_j = \tAlive$ and $\mpsMove^{t'}
\neq j$. Though, such a run would not satisfy the B\"{u}chi condition
in (\ref{buchi:dead}) leading to a contradiction.

\paragraph{\textbf{Case 3}: $\tRet$-$\tNoRet$ tag differs.}
Let $t$ denote the least such position.
\begin{claim}
  Let $j$ be in $\mpsMoves$. Let $\artag^t_j = \tNoRet$, then for all
  $t' > t$ where $\size{\mpsWord^{t'}_j} = \size{\mpsWord^{t}_j}$,
  $r^{t'}_j = \tNoRet$.
\end{claim}
\begin{proof}
  We will prove the statement for the smallest $t'$ with $t' > t$, and
  the claim will follow by induction on time.
  We consider the case where $j = \mpsMove^t$, as otherwise $t' = t+1$
  and the statement follows from condition \ref{trans:nochageRet}. If $\anaction^t =
  \actInt$, again $t' = t + 1$ and the statement follows from
  condition \ref{trans:retTagInt}. $\anaction^t = \actRet$ is not possible because of
  condition \ref{trans:retTagRet}. Hence, the interesting case is when $j = \mpsMove^t$
  and $\anaction^t = \actCall$. Then if $t'$ is as defined above, we
  can conclude from semantics of the multi-pushdown systems that $s^{t'-1} =
  j$ with $\anaction^{t'-1} = \actRet$. Further, the character
  ``popped'' from the stack at time $t'-1$ is the same as the one that
  was ``pushed'' at time $t$ -- which was $\tNoRet$ (condition \ref{trans:retTagCall}). This
  in turn shows, that $\artag^{t'}_j = \tNoRet$ (condition \ref{trans:retTagRet}).
\end{proof}
The following claim which rules out the tag being incorrectly marked
$\tNoRet$ follows from the previous claim, along with condition 6 that
disallows the action to be $\actRet$ when tag is $\tNoRet$.
\begin{claim}
  Let $j$ in $\mpsMoves$. If $\artag^t_j = \tNoRet$, then $\forall t'
  > t$, $\size{\mpsWord^{t'}_j} \geq \size{\mpsWord^t_j}$.
\end{claim}
Next, we consider the case when the tag is marked $\tRet$ and show
there is indeed a position in the future when the stack returns.
\begin{claim}
  Let $j \in \mpsMoves$ and $\artag^t_j = \tRet$. Then, $\exists t'
  > t$ such that $\size{\mpsWord^{t'}_j} < \size{\mpsWord^t_j}$.
\end{claim}
\begin{proof}
  It is easy to show that while the height of the stack is more than
  that at time $t$ the tag will be $\tRet$ (using constraints
  \ref{trans:retTagCall}-\ref{trans:retTagRet}).
  In case this stack never becomes dead, the B\"uchi condition (\ref{buchi:return})
will not be satisfied.
  In the case the stack eventually becomes dead, condition
  \ref{trans:deadEmptyAtom} in the transition relation which requires tag to
  be $\tNoRet$ when stack goes dead leads to contradiction.
\end{proof}

\paragraph{\textbf{Case 4}: Atom differs.} Let $\aformulabis$ denote the
smallest formula on which the atoms differ. Let $t$ be the smallest
position where the difference is for $\aformulabis$. That is, for some
$j\in \mpsMoves$, $A^t_j$ does not match definition obtained for atoms
for $\augmentrun{\arun}$ in (\ref{sync:cond:atom}) with respect to
presence (or absence) of $\aformulabis$.
$j = s^{t-1}$, for otherwise it is easy to conclude using constraint
\ref{trans:nochangeAtom} (which states only the atom for an active
stack can change) that one may find a smaller $t$ where the runs
differ with respect to $\aformulabis$. Let $t^* \geq t$ be the smallest
such that $s^{t^*} = j$. If such a $t^*$ does not exist then the stack
would in fact be dead, and constraint \ref{trans:deadEmptyAtom} will
imply that the atom is the empty set -- matching definition of atom
for $\augmentrun{\arun}$, a contradiction. Thus, we have $\aformulabis$
present (respectively, absent) in atom $\atom^{t^*}_j$ but $\arun, t^*
\not\models \aformulabis$ (respectively, $\arun, t^* \models
\aformulabis$). To conclude
the previous statement, we made use of constraint
\ref{trans:nochangeAtom} and definition of atom for
$\augmentrun{\arun}$ in (\ref{sync:cond:atom}).

The rest of the analysis depends of the type of formula $\aformulabis$
is.  We consider the case when $\aformulabis$ is an abstract-until
formula in detail. The other cases for the atom are much simpler we
point to parts of transition relation constraints
($\vec{\prodify{\mpsTrans}}$ in Figure \ref{figure:newTrans}) and
B\"uchi acceptance conditions ($\mathcal{F}$) required to reach a
contradiction. Correctness of atomic propositions is ensured by
constraint \ref{trans:stateConsistent}, that of propositional parts of
the formulas by definition of atom, that of propagating next formulas
by constraint \ref{trans:globalnext}, and for abstract-next by
constraints
\ref{trans:callHandleAbstract}-\ref{trans:returnStackAtom}, until
formulas by atom and B\"uchi acceptance sets (\ref{buchi:until}), and
finally caller formulas (which are easier to handle since they are
only passed down the stack) by constraints
\ref{trans:pushHandleCaller}-\ref{trans:abstractUntilPropagateInactive}. We
also do not elaborate on the case of negation of abstract-until is
incorrectly present (i.e. \emph{release}) which can be established
with the help of transition constraints
\ref{trans:callHandleAbstract}-%
\ref{trans:obligFree} and
\ref{trans:abstractUntilPropagateActive}-%
\ref{trans:abstractUntilPropagateInactive}.

%

%
\begin{figure}
\includegraphics[width=\textwidth]{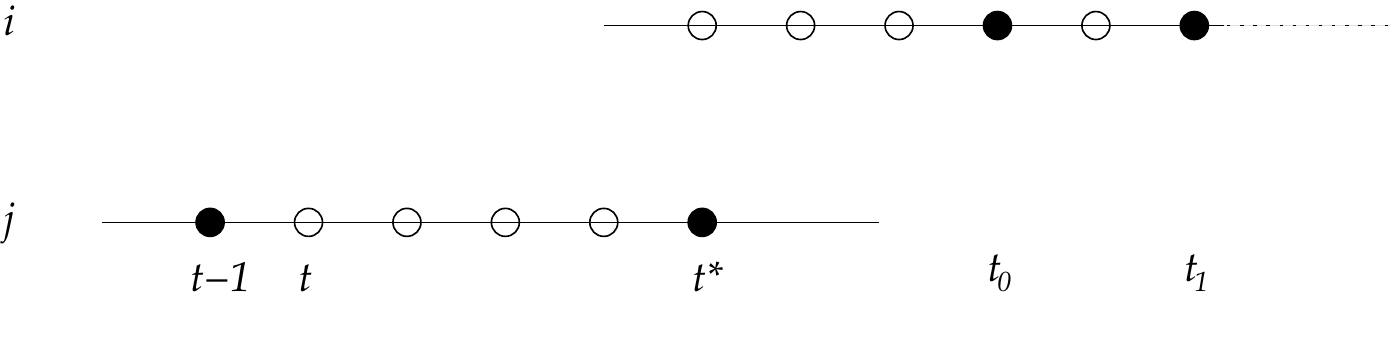}
\caption{Visual representation of $t$-s in proof for abstract-until.
Solid circles correspond to the particular stack being active.}
\end{figure}
For the rest of the proof we focus on the subcase where
$\aformulabis = \aformula_1 \until^\opAbs_i \aformula_2$
such that $\aformulabis \in
\atom^{t^*}_j$ but $\arun, t^* \not\models \aformulabis$.
%
%
Let $t_0 \geq t^*$ be smallest such with $s^{t_0}$ = $i$.  We can
conclude that $\aformulabis \in \atom^{t_0}_i$ (trivially if $i=j$, and
using constraint \ref{trans:abstractUntilPropagateInactive} and
\ref{trans:nochangeAtom} otherwise), and $\arun, t_0 \not\models
\aformulabis$ (semantics of operator $\until^\opAbs_i$).
Define $t_1$, $t_2$, $\ldots$ as sequence of positions with $t_{k+1} =
\Succ^{\opAbs,i}_{\arun}(t_k)$.
From constraint \ref{trans:abstractUntilPropagateActive} about
abstract-until formula for active stack in atom corresponding to
active stack, if $\aformulabis \in \atom^{t_0}_i$, then either
$\aformula_2 \in \atom^{t_0}_i$ or $\aformula_1 \in \atom^{t_0}_i$ and
$\next^\opAbs_i \aformulabis \in \atom^{t_0}_i$. Using the fact that
$\aformulabis$ is the smallest ``incorrect'' formula we rule out the
former possibility (it will imply $\arun,  t_0 \models \aformula_2$
-- a contradiction to $\arun, t_0 \not\models \aformulabis$).
Thus, $\arun, t_0 \models \aformula_1$
and $\next^\opAbs_i \aformulabis \in \atom^{t_0}_i$. From constraints
\ref{trans:callHandleAbstract}-\ref{trans:returnStackAtom} it follows
that $\aformulabis \in \atom^{t_1}_i$. Inductively, one may show that
$\aformulabis$ is in $\atom^{t_0}_i$, $\atom^{t_1}_i$, $\atom^{t_2}_i$
etc., and hence also $\arun, t_0 \models \aformula_1$,
$\arun, t_1 \models \aformula_1$ etc.

We have
established if an abstract-until formula is in an atom, the
first part of the until formula is true for successive-abstract
positions for the particular stack it is asserted on. The rest of the
argument would have proceeded analogous to the standard until formula,
by having a B\"uchi acceptance condition to ensure that eventually the
second part of the formula holds. As such, the argument does not work since
the sequence of
successive positions is not guaranteed to be infinite. We need to
ensure that we have satisfied the until formula by the last position
through the transition relation. To be able to do that, we need to be able
to identify these last positions locally. This is where the tags will
be helpful -- they help us exactly determine the cases when the
sequence of successive positions where abstract until formula is to be
asserted is finite and the case when it is infinite and we require
a B\"uchi acceptance set. We now present a more formal treatment, which will
illustrate some subtleties (or technicalities, as one may prefer) of
the argument.

First of all,  we recall that from previous arguments we already know
the tags are correct (that is, they follow the semantics as in
(\ref{sync:cond:dead}), (\ref{sync:cond:return})). The sequence would
be finite if $\Succ^{\opAbs,i}_\arun(t_k)$ does not exist for some $k$. We do
a case analysis on the action performed on the stack.
If the action $\anaction^{t_k}$ is $\actCall$, we consider two cases
when the abstract-successor may not exist. First case is when the call
does not return at all, that is $\forall t' > t_k$,
$\size{\mpsWord^{t'}_i} > \size{\mpsWord^{t_k}_i}$. In other words,
$\forall t' \geq t_k + 1$, $\size{\mpsWord^{t'}_i} \geq
\size{\mpsWord^{t_k+1}}$. In such a case $\artag^{t_k+1}_i$ will be
$\tNoRet$ and hence because of constraint
\ref{trans:noreturnObligFree} there will be no abstract-next
obligations. Second case is when the call does return. Even in such a
case the successor may not exist if stack is dead at this
point. Formally, if $t' > t$ denotes the least position such that
$\mpsWord^{t'}_i = \mpsWord^{t_k}_i$, successor will not exist if
$\forall t'' \geq t'$, $s^{t''} \neq i$. In other words,
$\adtag^{t'}_i = \tDead$. Since this is the least such position, we
can conclude that $s^{t'-1} = i$ with $\anaction^{t'-1} = \actRet$
with the character being popped the one that was pushed at time $t_k$.
Since a dead stack will not have any obligations (constraint
\ref{trans:deadEmptyAtom}), we may conclude from constraint
\ref{trans:returnStackAtom} that the atom pushed at time $t_k$ has no
abstract-next formula. To conclude, $\atom^{t_k}_i$ cannot have any
abstract-next obligations if $\anaction^{t_k} = \actCall$ and
$\Succ^{\opAbs,i}_\arun(t_k)$ is undefined.
If the action $\anaction^{t_k}=\actRet$, constraint
\ref{trans:popObligFree} ensures the same.
Finally, for the case $\anaction^{t_k}=\actInt$, abstract-successor
will not exist only if stack becomes dead, which is handled by
constraint \ref{trans:dyingObligFree}.

On the contrary, let us consider the case the sequence of
abstract-successors is infinite. It must be the case that at each of
these positions the tags are $\tNoRet$ (follows from semantics of
abstract-successor and definition of tags). We also observe that for
each until formula there might be at most one such infinite
sequence. Together these observations would imply the run does not
satisfy the B\"uchi acceptance conditions (\ref{buchi:absuntil}).

We have exhausted all possible cases and may safely claim that,
indeed, $\aprodrun = \augmentrun{\arun}$.

%% file: subsection-repreach.tex
\section{Complexity Analysis with Bounded Runs}

\subsection{Bounded Repeated Global State Reachability Problem}
In this section, we evaluate the computational complexity of the
problem BREP as well as its variant restricted to a single accepting global
state, written below BREP$_{\rm single}$. First, note that there is a 
logarithmic-space 
reduction
from BREP to BREP$_{{\rm single}}$ by copying the  multi-pushdown  system as many times as
the cardinality of $\mathcal{F}$ (as done to reduce non-emptiness problem for
generalized B\"uchi automata to non-emptiness problem for 
\ifLONG standard \fi B\"uchi automata). 
This allows us to conclude about the 
complexity upper bound for BMC itself but it is worth noting 
that the multi-pushdown system obtained by synchronization has an exponential number of global states
and therefore a refined complexity analysis is required to get optimal upper bounds.

In order to analyze the complexity for BREP$_{{\rm single}}$, we take advantage of two proof techniques
\ifLONG that have been introduced earlier and \fi 
for which we provide a complexity analysis that will suit our final
goal. Namely, existence of an infinite $k$-bounded run such that 
\ifLONG a final global state \fi 
 $\pair{\mpsState_f}{i_f}$ is repeated infinitely often is checked: 
\ifLONG
\begin{enumerate}
\itemsep 0 cm
\item[(1)] by  first guessing a sequence of 
intermediate global states witnessing context switches of length at most $k+1$,
\item[(2)] by  computing the (regular) set of reachable configurations following that sequence and then,
\item[(3)] by verifying whether 
 there is a reachable configuration leading to an infinite run such that  $\pair{\mpsState_f}{i_f}$ 
is repeated infinitely often
and no context switch occurs.
\end{enumerate}
\fi 
 The principle behind  (2) is best explained in~\cite{Qadeer&Rehof05} but we 
provide a complexity analysis using the computation of post$^{\star}(\aset)$ along the lines of~\cite{Schwoon02}.
Sets post$^{\star}(\aset)$  need to be computed at most $k$ times, which might cause an 
exponential blow-up (for instance if at each step the number
of states were multiplied by a constant).
Actually,  computing post$^{\star}$ adds an additive factor at each step, which is
essential for our complexity analysis.
\ifLONG
Let us define the problem BREP$_{\rm single}$ (bounded repeated global state reachability problem for multi-pushdown 
systems):
\begin{description}
\itemsep 0 cm
\item[input:] a multi-pushdown system $\mps$, an initial configuration $\pair{\pair{\astate}{i}}{(\bottom)^N}$,
a global state $\pair{\astate_f}{i_f}$ and a bound $k \in \Nat$.
\item[question:] Is there an infinite $k$-bounded run $\arun$ from 
$\pair{\pair{\astate}{i}}{(\bottom)^N}$ such that $\pair{\astate_f}{i_f}$ is repeated infinitely often?
\end{description}
\fi 

\begin{proposition} \label{proposition-breps}
BREP$_{\rm single}$ can be solved in time $\mathcal{O}(\length{\mps}^{k+1} \times p(k,\length{\mps}))$
for some polynomial $p(\cdot,\cdot)$. 
\end{proposition} 

The proof of Proposition~\ref{proposition-breps}   is at the heart of our complexity analysis and it relies on constructions 
from~\cite{Bouajjani&Esparza&Maler97,Schwoon02}. Moreover, we shall take advantage of it
when the input system is precisely $\prodmps$. 

\input{proof-proposition-breps}

\begin{corollary} \label{corollary-k-boundedness}
(I) BMC with $k$ encoded with a unary representation is \exptime-complete.
(II) BMC 
\ifLONG with $k$ encoded with a binary representation
\fi  is in 2\exptime. 
\end{corollary}

\input{proofsketch-corollary-k-boundedness}

It is worth also noting that~\cite[Theorem 15]{BCGZ11} would  lead to an \exptime \ upper
bound for BMC if $k$ is not part of the input, see the \exptime \ upper bound
for the problem NESTED-TRACE-SAT($\mathcal{L}^{-},k$) introduced in~\cite{BCGZ11}; 
in our case $k$ {\em is} indeed part of the input and in that
case, the developments in~\cite{BCGZ11}  will lead to a 2\exptime \ bound by using the method used for   NESTED-TRACE-SAT($\mathcal{L}^{-},k$)
even if $k$ is encoded in unary.
Indeed, somewhere in the proof, the path expression $succ_{\leq k}$ is exponential in the value $k$.
Hence, Corollary~\ref{corollary-k-boundedness}(I) is the best we can hope for when $k$ is part of the input of
the model-checking problem. 

Adding regularity constraints about stack contents preserves the complexity upper bound.
We write BMC$^{\sc reg}$ to denote the extension of BMC in which Multi-\caret \ is replaced
by Multi-\caret$^{\sc reg}$ (see Section~\ref{section-definition-logic}). 

\begin{corollary} \label{corollary-k-boundedness-reg}
(I) BMC$^{\sc reg}$ with $k$ encoded with an unary representation is \exptime-complete.
(II) BMC$^{\sc reg}$ 
\ifLONG with $k$ encoded with a binary representation
\fi  
is in 2\exptime. 
\end{corollary}

\input{proofsketch-regularity}

%% file: proof-proposition-breps.tex
\begin{proof} We use the following results on pushdown systems.
Let $\mps = (\mpsStates,  \mpsAlphabet, \mpsTrans)$  be a pushdown system and
$\aautomaton$ be a $\mps$-automaton encoding a regular set of configurations. 
We recall that a  
\defstyle{$\mps$-automaton} $\aautomaton = \triple{Q}{\mpsAlphabet}{\delta,\mpsStates,F}$
is a finite-state automaton over the alphabet $\mpsAlphabet$ such that $\mpsStates
\subseteq Q$~\cite{Schwoon02}.
A configuration $\pair{\astate}{\aword}$ is \defstyle{recognized} by  $\aautomaton$ (written 
 $\pair{\astate}{\aword} \in \alang(\aautomaton)$) $\equivdef$ 
$\astate \step{\aword} \astate'$ in $\aautomaton$ for some accepting state $\astate' \in F$.
Hence, the set $\mpsStates$ in $\triple{Q}{\mpsAlphabet}{\delta,\mpsStates,F}$ can be viewed as a set of initial states and acceptance
is relative to the initial state $\astate$ that allows to satisfy $\astate \step{\aword} \astate'$.
Note that a $\mps$-automaton is nothing else than a means to represent a regular set of configurations.

Let us list a few essential properties.
\begin{enumerate}
\itemsep 0 cm 
%
%
\item[(i)] 
Let $\poststar{\aautomaton} 
\egdef
\set{
\pair{\astate}{\aword}:
\exists \pair{\astate'}{\aword'} \in \alang(\aautomaton) \
{\rm s.t.} \  \pair{\astate'}{\aword'} \step{*} \pair{\astate}{\aword} \ 
{\rm in} \ \mps
}$.
The set
$\poststar{\aautomaton}$ can be represented by a $\mps$-automaton $\aautomaton'$
such that the number of states for $\aautomaton'$ is bounded
by the number of states for $\aautomaton$ plus $\card{\mpsStates} \times  \card{\mpsAlphabet}$
and the time required to compute  $\aautomaton'$ is quadratic in its number of states and in 
$\card{\mpsAlphabet}$. 
Clearly, $\aautomaton'$ can be computed
in polynomial time in $\length{\aautomaton} + \length{\mps}$, see e.g.~\cite{Schwoon02} but we shall also take
advantage of the fact that 
the number of states is at most augmented by a constant factor~\cite[Section 3.3.3]{Schwoon02}. 

\item[(ii)] Checking whether there is an infinite run $\arun$ starting from a configuration
in $\alang(\aautomaton)$ and such that the global state $\pair{\mpsState_f}{i_f}$ is repeated
infinitely often, can be done in polynomial time in $\length{\aautomaton} + \length{\mps}$,
see e.g.~\cite{Bouajjani&Esparza&Maler97}. 

\end{enumerate}

Let $\mps = 
 (\mpsStates \times \insqrbr{\mpsStackCount}, \mpsStackCount, \mpsAlphabet, \mpsTrans_1, 
 \ldots ,\mpsTrans_{\mpsStackCount})$ be an enhanced multi-pushdown system, $\pair{\astate}{i}$ and $\pair{\astate_f}{i_f}$ be global states and
$k \in \Nat$. For every $\couple{\astate_1}{i_1} \cdots \couple{\astate_l}{i_l} \in (\mpsStates \times \insqrbr{\mpsStackCount})^*$ 
such that $l \leq k+1$, $\couple{\astate}{i} = \couple{\astate_1}{i_1}$ and $i_f = i_l$, we check
whether there is an infinite run $\arun$ from 
$\couple{\couple{\astate}{i}}{(\bottom)^N}$ such that 
\begin{enumerate}
\itemsep 0 cm 
\item $\pair{\mpsState_f}{i_f}$ is repeated infinitely often,
\item the sequence of active stacks in $\arun$ is exactly $i_1 \cdots i_l$
      and $\couple{\astate_1}{i_1} \cdots \couple{\astate_l}{i_l}$ witnesses which are the intermediate global states that there is a context switch.
\end{enumerate}
We show that the existence of such a run can be
done in polynomial time, which entails an \exptime \ upper bound since there is an exponential amount of
sequences of the form $\couple{\astate_1}{i_1} \cdots \couple{\astate_l}{i_l}$ with $l \leq k+1$.

The algorithm has two steps. First, we build an automaton $\aautomaton$ encoding
a (regular) set of configurations from 
$\triple{\mpsStates \times \insqrbr{\mpsStackCount}}{
\mpsAlphabet}{\mpsTrans_{i_l}}$ corresponding to the configurations
of $\mps$ restricted to the $i_l$-th stack that can be reached from  $\pair{\couple{\astate_1}{i_1}}{(\bottom)^N}$
via the sequence  $\couple{\astate_1}{i_1} \cdots \couple{\astate_l}{i_l}$. This is precisely  the approach followed
in~\cite{Qadeer&Rehof05} by picking a sequence of context switches and doing a $\poststar{\cdot}$
in that order to get the set of all reachable configurations.
We shall see that $\aautomaton$ can be indeed computed
in polynomial time in $\length{\mps}$. Then, we use the polynomial time algorithm from (ii) above
to check whether  there is an infinite run $\arun$ for the pushdown system
$\triple{\mpsStates \times \set{i_l}}{\mpsAlphabet}{\mpsTrans_{i_l}}$ that starts from a configuration
in $\alang(\aautomaton)$ and such that the global state $\pair{\mpsState_f}{i_f}$ is repeated
infinitely often. 
It remains to check that $\aautomaton$ can be computed
in polynomial time in $\length{\mps}$.
  
Let us introduce some notation. We write $\mps_j$ to denote the pushdown system 
$\triple{\mpsStates \times \insqrbr{\mpsStackCount}}{
\mpsAlphabet}{\mpsTrans_{j}}$. Note that $\mps$ and $\mps_j$ have identical sets of global states:
$\mps_j$ corresponds to $\mps$ restricted to the transition relation $\mpsTrans_j$. 
Given a $\mps$-automaton or a  $\mps_j$-au\-to\-ma\-ton $\aautomaton = 
 \triple{Q}{\mpsAlphabet}{\delta,\mpsStates \times \insqrbr{\mpsStackCount},F}$, we write FSA$(\aautomaton,\pair{\astate}{i})$
to denote a finite-state automaton such that for every $\aword \in \mpsAlphabet^*$, we have
$\aword \in \alang({\rm FSA}(\aautomaton,\pair{\astate}{i}))$ iff $\pair{\pair{\astate}{i}}{\aword} \in \alang(\aautomaton)$.
Note that  FSA$(\aautomaton,\pair{\astate}{i})$ can be easily obtained from $\aautomaton$ 
by replacing $\pair{\astate}{i}$ by a new state $q_0$ that is also the initial state.
Similarly,  given a finite-state automaton $\aautomaton$ over the alphabet $\mpsAlphabet$, we write
PA$(\aautomaton,\pair{\astate}{i})$ to denote a  $\mps$-automaton such that
\begin{enumerate}
\itemsep 0 cm
\item for every $\pair{\astate'}{i'} \neq \pair{\astate}{i}$, for  every $\aword \in \mpsAlphabet^*$, $\pair{\pair{\astate'}{i'}}{\aword} \not \in 
\alang({\rm PA}(\aautomaton,\pair{\astate}{i}))$,
\item  for  every $\aword \in \mpsAlphabet^*$, $\pair{\pair{\astate}{i}}{\aword} \in 
\alang({\rm PA}(\aautomaton,\pair{\astate}{i}))$ iff $\aword \in \alang(\aautomaton)$. 
\end{enumerate}
Without any loss of generality, we can assume that $\mpsStates  \times \insqrbr{\mpsStackCount}$ is disjoint from the set of states of $\aautomaton$ (otherwise, we rename
the states). PA$(\aautomaton,\pair{\astate}{i})$ can be obtained from $\aautomaton$ by adding all the  states from 
$\mpsStates  \times \insqrbr{\mpsStackCount}$ and by replacing in the transition relation, the initial state of $\aautomaton$ 
by $\pair{\astate}{i}$. 
Observe that the operations $\poststar{\cdot}$, FSA$(\cdot)$ and PA$(\cdot)$ define automata where the increase in the number of states
is bounded by $\card{\mpsStates} \times N$. Hence, performing such operations at most 
$3 \times (k+1)$ times, irrespective of the ordering of these operations,
will increase the number of states by at most $\card{\mpsStates} \times N \times 3 \times (k+1)$. 

Now, let us define the (simple) algorithm that consists in computing a $\mps_{i_l}$-automaton
that represents the set of configurations with global state $\pair{\astate_{l}}{i_l}$ reachable
from the initial configuration and following the sequence of intermediate global states (that witness
context switches too).

Let $\alpha := 1$ and define the finite-state automata $\aautomaton_1$, \ldots, $\aautomaton_N$ over the alphabet $\mpsAlphabet$
such that $\alang(\aautomaton_1) = \cdots = \alang(\aautomaton_N) = \set{\bottom}$. While $\alpha \leq l$,
perform the following
steps:
\begin{enumerate}
\itemsep 0 cm
\item Compute a  $\mps_{i_{\alpha}}$-automaton $\aautomatonbis$ that represents all the configurations reachable
      from a configuration of the form $\pair{\pair{\astate_{\alpha}}{i_{\alpha}}}{\aword}$ with $\aword \in \alang(\aautomaton_{i_{\alpha}})$ using
      only the stack $i_{\alpha}$. \\
     That is, $\aautomatonbis  := \poststar{{\rm PA}(\aautomaton_{i_{\alpha}}, \pair{\astate_{\alpha}}{i_{\alpha}})}$ (computed from $\mps_{i_{\alpha}}$);
\item If $\alpha <  l$, then update $\aautomaton_{i_{\alpha}}$ so that it represents the set of contents for the stack $i_{\alpha}$ from the configurations
      represented by $\aautomatonbis$. Otherwise return 
      the $P$-automaton 
      ${\rm PA}(\aautomaton_{i_{l}}, \pair{\astate_{l}}{i_{l}})$.\\
      That is, if $\alpha <  l$ then $\aautomaton_{i_{\alpha}} := {\rm FSA}(\aautomatonbis,  \pair{\astate_{\alpha+1}}{i_{\alpha+1}})$
      else return ${\rm PA}(\aautomaton_{i_{l}}, \pair{\astate_{l}}{i_{l}})$.
\item $\alpha := \alpha +1$.
\end{enumerate}
Let the $\mps_{i_l}$-automaton returned by the above algorithm be denoted by $\aautomaton$.
Then, as explained earlier, we check whether  there is an infinite run $\arun$ for 
$\triple{\mpsStates \times \set{i_l}}{\mpsAlphabet}{\mpsTrans_{i_l}}$ that starts from a configuration
in $\alang(\aautomaton)$ and such that $\pair{\mpsState_f}{i_f}$ is repeated
infinitely often.

Note that $\aautomaton$ is obtained from automata with a few states after applying the operations 
$\poststar{\cdot}$, FSA$(\cdot)$ and PA$(\cdot)$ at most $3 \times (k+1)$ times.
Hence, the size of $\aautomaton$ is in 
$\mathcal{O}([3 \times (k+1) \times \card{\mpsStates} \times N]^2 \times \card{\mpsAlphabet})$.
Detecting whether there is an infinite run in which  $\pair{\mpsState_f}{i_f}$ is repeated infinitely often,
will be polynomial in $3 \times (k+1) \times \card{\mpsStates} \times N \times \card{\mpsAlphabet}$. 
As a consequence, the complete decision procedure requires time in $\mathcal{O}(\length{\mps}^{k+1} \times p(k,\length{\mps}))$
for some polynomial $p(\cdot,\cdot)$. 
\end{proof}

%% file: proofsketch-corollary-k-boundedness.tex
In Corollary~\ref{corollary-k-boundedness}(I), \exptime-hardness is inherited from the 
case with a single stack~\cite{Bouajjani&Esparza&Maler97}. 
We have seen that there is  an infinite $k$-bounded run $\arun$ from 
$\pair{\astate}{(\bottom)^N}$ such that $\arun \models \aformula$ iff
there exists a $k$-bounded run $\widehat{\arun}$ from
$(\widehat{g}, (\bottom)^N)$ for some $\widehat{g} \in I_0$
from some  multi-pushdown system $\prodmps$
 such that for each $F \in \mathcal{F}$ there exists a $\mpsState_f
    \in F$ that is repeated infinitely often.
Since  $\prodmps$, $I_0$ and $\mathcal{F}$ are of exponential size, the second proposition
can be reduced to an exponential number of instances of BREP$_{\rm single}$ in which the multi-pushdown
system is of exponential-size only. Using the parameterized complexity upper bound  
$\mathcal{O}(\length{\mps}^{k+1} \times p(k,\length{\mps}))$, we can conclude that BMC with 
$k$ encoded with an unary representation can be solved in exponential time.  
Corollary~\ref{corollary-k-boundedness}(II) is then a consequence of the above argument.

%% file: proofsketch-regularity.tex
Let us explain how the construction of $\prodmps$ can be updated so that BMC$^{\sc reg}$
can be solved almost as BMC. Obviously, we have to take care of regularity constraints and to do so,
we  enrich the global states with pieces of information about the regularity constraints
satisfied by the current stack content. Typically, such pieces of information shall be finite-state automata
enriched with a set of states and an update on a stack triggers an update on the set of states.
Moreover, we  take advantage of the stack mechanism to recover previous values of such pieces of information.

Let us provide a bit more detail. Suppose that the formula $\aformula$ contains the following
regularity constraints ${\tt in}(s_1, \aautomaton_1)$, \ldots, ${\tt in}(s_n, \aautomaton_n)$. We extend the notion
of augmented run so that global states are enriched with triples 
$(s_1,\aautomaton_1, \aset_1)$, \ldots, $(s_n,\aautomaton_1, \aset_n)$ where each $\aset_i$ is a set of states
from $\aautomaton_i$. By definition, $\aset _i$ is the set of states from $\aautomaton_i$ that can be reached
from some initial state of $\aautomaton_i$ with the current content of the stack $s_i$. Hence, $\aset_i$ is uniquely
defined but since $\aautomaton_i$ is not necessarily deterministic, $\aset_i$ may not be a singleton. 
In the definition of $\prodmps$, $(s_1,\aautomaton_1, \aset_1)$, \ldots, $(s_n,\aautomaton_1, \aset_n)$ are updated
according to the updates on stacks but we have to be a little bit careful. Before explaining the very reason,
first note that if ${\tt in}(s_i, \aautomaton_i)$ belongs to an atom of some global state of  $\prodmps$,
we impose that $\aset_i$ contains an accepting state of $\aautomaton_i$, which amounts to check that
the current content of the stack $s_i$ is indeed a pattern from $\alang(\aautomaton_i)$. 
Let us explain now how to update the values  $(s_1,\aautomaton_1, \aset_1)$, \ldots, $(s_n,\aautomaton_1, \aset_n)$.
When a call action is performed on a stack $s$ with letter $a$, 
each value  $(s_i,\aautomaton_i, \aset_i)$ with $s_i = s$ is replaced by
 $(s_i,\aautomaton_i, \asetbis_i)$ where $\asetbis_i$ is the set of states that can be reached
from some state of $\aset_i$ by reading $a$ (as what is done in the power set construction for
finite-state automata). 
Moreover, we extend the stack alphabet of $\prodmps$ so that each letter of the stack alphabet is also
enriched with values  $(s_1,\aautomaton_1, \aset_1)$, \ldots, $(s_n,\aautomaton_1, \aset_n)$.
When a call action is performed in $\mps$, we perform also a call in  $\prodmps$, but with a letter enriched
with the values $(s_1,\aautomaton_1, \aset_1)$, \ldots, $(s_n,\aautomaton_1, \aset_n)$ on the stack.
Now, when a return is performed on a stack $s$, the current values 
 $(s_1,\aautomaton_1, \aset_1)$, \ldots, $(s_n,\aautomaton_1, \aset_n)$ on the top of the stack
are used to restore those values in the global state of $\prodmps$.
Similarly, when an internal action is performed on a stack $s$, 
the current values 
 $(s_1,\aautomaton_1, \aset_1)$, \ldots, $(s_n,\aautomaton_1, \aset_n)$ on the top of the stack
are also used to get the new values in the global state (but this time these values are not popped from
the stack). By observing that values of the form $(s_1,\aautomaton_1, \aset_1)$, \ldots, $(s_n,\aautomaton_1, \aset_n)$ 
are of linear size in the size of $\aformula$, all the complexity analysis we have performed
for BMC can be easily adapted  to BMC$^{\sc reg}$; actually, the arguments are identical except that the construction
of  $\prodmps$ is a bit more complex, as described above.

%% file: section-others.tex
\subsection{Complexity Results for Other  Boundedness Notions}
\label{section-others}

\ifLONG
In this section, we focus on the complexity analysis for OBMC and PBMC based not only
on previous developments but also on the complexity of repeated reachability problems
when runs are either $k$-phase-bounded or from ordered multi-pushdown systems. 
\fi 
Let OREP$_{\rm single}$ be the variant of BREP$_{\rm single}$ with ordered 
multi-pushdown systems: 
\ifLONG
\begin{description}
\itemsep 0 cm
\item[input:] an ordered multi-pushdown system $\mps$, a 
configuration $\pair{\pair{\astate}{i}}{(\bottom)^N}$,
a global state $\pair{\astate_f}{i_f}$.
\item[question:] Is there an infinite  run $\arun$ from 
$\pair{\pair{\astate}{i}}{(\bottom)^N}$ such that $\pair{\astate_f}{i_f}$ is repeated infinitely often?
\end{description}
\fi 
According to~\cite[Theorem 11]{Atig10bis}, 
OREP$_{\rm single}$ restricted to  ordered multi-pushdown systems with $k$ stacks
can be checked in time
$\mathcal{O}(\size{\mps}^{2^{d \ k}})$ where $d$ is a constant.
Our synchronized product $\prodmps$ is exponential in the size of formulas 
(see Section~\ref{section-synchronisation}), whence
order-bounded model-checking problem OBMC can be solved in
$2 \exptime$ too ($k$ is linear in the size of our initial $\mps$). 
Note that Condition (C) from Theorem~\ref{theorem:mc2rep} needs to be used here.

\begin{corollary} \label{corollary-obmc}
OBMC  is in $2 \exptime$.
\end{corollary}

Corollary~\ref{corollary-obmc} is close to optimal since
non-emptiness problem for ordered multi-pushdown systems is 2\etime-complete. 
In addition, the same complexity upper bounds apply even when regularity constraints on stack
contents are added. 

Now, let us conclude this section by considering $k$-bounded-phase runs. Again,
let us define the problem PBREP$_{\rm single}$:
\ifLONG
\begin{description}
\itemsep 0 cm
\item[input:] a multi-pushdown system $\mps$, an initial configuration $\pair{\pair{\astate}{i}}{(\bottom)^N}$,
a global state $\pair{\astate_f}{i_f}$ and a bound $k \in \Nat$.
\item[question:] Is there an infinite $k$-phase-bounded run $\arun$ from 
$\pair{\pair{\astate}{i}}{(\bottom)^N}$ such that $\pair{\astate_f}{i_f}$ is repeated infinitely often?
\end{description}
\fi
In~\cite[Section 5]{Atig&Bollig&Habermehl08}, it is shown that non-emptiness
for $k$-phase multi-pushdown systems can be reduced to
non-emptiness for ordered multi-pushdown systems with $2k$ stacks. By inspecting the proof, we can conclude:
\ifLONG
\begin{enumerate}
\itemsep 0 cm
\item a similar reduction can be performed for reducing the repeated reachability of a global state,
\item non-emptiness of $k$-phase $\mps$ with $N$ stacks is reduced to 
non-emptiness of one of $N^{k}$ instances of $\mps'$ with $2k$ stacks and each $\mps'$ is
polynomial-size in $k + \length{\mps}$.
\end{enumerate}
\fi 
Therefore, PBREP$_{\rm single}$ is in  $2 \exptime$ too.
Indeed, there is an exponential number of instances and checking non-emptiness for one of them
can be done in double exponential time. 
By combining the different complexity measures above, checking an instance
of PBREP$_{\rm single}$ with  $\prodmps$ requires time in
\ifLONG 
$$
\mathcal{O}(
N^k \times \size{\prodmps}^{2^{d \ 2k}}
)
$$
\fi
which is clearly double-exponential in the size of $\mps$.  
Consequently, bounded model-checking with bounded-phase  multi-pushdown systems
is in $2 \exptime$ too if the number of phases is encoded in 
\ifLONG 
unary (and in $3 \exptime$ otherwise).
\fi 
\begin{corollary} \label{corollary-pbmc}
(I) PBMC where $k$ is encoded in unary is in $2 \exptime$.
(II) PBMC where $k$ is encoded in binary  is in $3\exptime$. 
\end{corollary}
Again, the same complexity upper bounds apply when regularity constraints  are added. 
Note that an alternative proof of Corollary~\ref{corollary-pbmc}(I)
can be found in the recent paper~\cite[Theorem 5.2]{Bollig&Kuske&Mennicke12} where fragments
of MSO are taken into account.

%% file: section-conclusion.tex
\section{Conclusion}
\label{section-conclusion}

In this note, we have shown that model-checking over multi-pushdown
systems with $k$-boun\-ded runs is \exptime-complete when $k$ is
\ifLONG an input bound \fi encoded in unary, otherwise the problem is
in 2\exptime \ with a binary encoding.
The logical \ifLONG specification \fi language is a version of \caret~in which abstract temporal
operators are related to calls and returns and parameterized by the
stacks, and regularity constraints on stack contents are present
too. A 2\exptime \ upper bound is also established with ordered
multi-pushdown systems or with $k$-phase bounded runs and these are
optimal upper bounds with a unary encoding of $k$. Our complexity
analysis rests on the reduction from model-checking to repeated
reachability and on complexity analysis for pushdown systems. 
%
The characterization of the complexity when $k$ is
encoded in binary is still open and we conjecture that an exponential
blow-up may occur.  More generally, our work can be pursued in several
directions including refinements of the complexity analysis but also
increase of the expressive power of the logics while keeping the same
worst-case complexity upper bounds.

%% file: section-appendix.tex
\section{Proof sketch for Lemma \ref{lemma:soundness}}
\label{section-appendix-soundness-lemma}
\begin{proof}[Sketch]\draftnote{Do we want to call this a sketch?}
We observe below properties about the way tags are placed in $\augmentrun{\arun}$.

\input{movable-constructedrun-semantics.tex}
Transition relations in $\widehat{\mps}$ and the acceptance conditions
in $\mathcal{F}$ mimic syntactically the semantical properties of the
augmented runs observed above, from which the lemma follows.
\end{proof}

\section{Proof sketch for Lemma \ref{lemma:completeness}}
\label{section-appendix-completeness-lemma}
\begin{proof}
\input{proof-completeness-lemma}
\end{proof}


%% file: techreport.bbl
\begin{thebibliography}{10}

\bibitem{Alur&Etessami&Madhusudan04}
R.~Alur, K.~Etessami, and P.~Madhusudan.
\newblock A temporal logic of nested calls and returns.
\newblock In {\em TACAS'04}, volume 2988 of {\em LNCS}, pages 467--481.
  Springer, 2004.

\bibitem{Atig10}
M.~Atig.
\newblock From multi to single stack automata.
\newblock In {\em CONCUR'10}, volume 6269 of {\em LNCS}, pages 117--131.
  Springer, 2010.

\bibitem{Atig10bis}
M.~Atig.
\newblock Global model checking of ordered multi-pushdown systems.
\newblock In {\em FST\&TCS'10}, pages 216--227. LIPICS, 2010.

\bibitem{Atig&Bollig&Habermehl08}
M.~Atig, B.~Bollig, and P.~Habermehl.
\newblock Emptiness of multi-pushdown automata is {2ETIME}-complete.
\newblock In {\em DLT'08}, volume 5257 of {\em LNCS}, pages 121--133. Springer,
  2008.

\bibitem{Atigetal12}
M.~Atig, A.~Bouajjani, K.~Kumar, and P.~Saivashan.
\newblock Model checking branching-time properties of multi-pushdown systems is
  hard.
\newblock Technical Report arXiv:1205.6928, arXiv:cs.LO, May 2012.

\bibitem{Atigetal12bis}
M.~Atig, A.~Bouajjani, K.~N. Kumar, and P.~Saivasan.
\newblock Linear-time model-cheking for multithreaded programs under
  scope-bounding.
\newblock In {\em ATVA'12}, volume 7561 of {\em LNCS}, pages 152--166.
  Springer, 2012.

\bibitem{Ball&Rajamani01}
T.~Ball and S.~Rajamani.
\newblock The {SLAM} {T}oolkit.
\newblock In {\em CAV'01}, volume 2102 of {\em LNCS}, pages 260--264. Springer,
  2001.

\bibitem{Biereetal03}
A.~Biere, A.~Cimatti, E.~M. Clarke, O.~Strichman, and Y.~Zhu.
\newblock Bounded model checking.
\newblock {\em Advances in Computers}, 58:118--149, 2003.

\bibitem{BCGZ11}
B.~Bollig, A.~Cyriac, P.~Gastin, and M.~Zeitoun.
\newblock Temporal logics for concurrent recursive programs: Satisfiability and
  model checking.
\newblock In {\em MFCS'11}, volume 6907 of {\em LNCS}, pages 132--144.
  Springer, 2011.

\bibitem{Bollig&Kuske&Mennicke12}
B.~Bollig, D.~Kuske, and R.~Mennicke.
\newblock The complexity of model-checking multi-stack systems.
\newblock 2012.
\newblock Submitted.

\bibitem{Bouajjani&Esparza&Maler97}
A.~Bouajjani, J.~Esparza, and O.~Maler.
\newblock Reachability analysis of pushdown automata: application to
  model-checking.
\newblock In {\em CONCUR'97}, volume 1243 of {\em LNCS}, pages 135--150.
  Springer, 1997.

\bibitem{Clarkeetal05}
E.~Clarke, D.~Kroening, N.~Sharygina, and K.~Yorav.
\newblock {SATABS}: {SAT}-based predicate abstraction for {ANSI-C}.
\newblock In {\em TACAS'05}, volume 3440 of {\em LNCS}, pages 570--574.
  Springer, 2005.

\bibitem{Courcelle90}
B.~Courcelle.
\newblock The monadic second order theory of graphs {I}: Recognisable sets of
  finnite graphs.
\newblock {\em I\&C}, 85:12--75, 1990.

\bibitem{Cyriac&Gastin&Kumar12}
A.~Cyriac, P.~Gastin, and K.~N. Kumar.
\newblock {MSO} decidability of multi-pushdown systems via split-width.
\newblock In {\em {CONCUR}'12}, volume 7454 of {\em LNCS}, pages 547--561.
  Springer, 2012.

\bibitem{Esparza&Ganty11}
J.~Esparza and P.~Ganty.
\newblock Complexity of pattern-based verification for multithreaded programs.
\newblock In {\em POPL'11}, pages 499--510. ACM, 2011.

\bibitem{Finkel&Willems&Wolper97}
A.~Finkel, B.~Willems, and P.~Wolper.
\newblock A direct symbolic approach to model checking pushdown systems.
\newblock In {\em INFINITY'97}, volume~9 of {\em ENTCS}, 1997.

\bibitem{Fischer&Ladner79}
M.~Fischer and R.~Ladner.
\newblock Propositional dynamic logic of regular programs.
\newblock {\em Journal of Computer and System Sciences}, 18:194--211, 1979.

\bibitem{Ibarra78}
O.~Ibarra.
\newblock Reversal-bounded multicounter machines and their decision problems.
\newblock {\em JACM}, 25(1):116--133, 1978.

\bibitem{Madhusudan&Parlato11}
P.~Madhusudan and G.~Parlato.
\newblock The tree width of auxiliary storage.
\newblock In {\em POPL'11}, pages 283--294. ACM, 2011.

\bibitem{Minsky67}
M.~Minsky.
\newblock {\em Computation, Finite and Infinite Machines}.
\newblock Prentice Hall, 1967.

\bibitem{Qadeer&Rehof05}
S.~Qaader and J.~Rehof.
\newblock Context-bounded model checking of concurrent software.
\newblock In {\em TACAS'05}, volume 3440 of {\em LNCS}, pages 93--107.
  Springer, 2005.

\bibitem{Schwoon02}
S.~Schwoon.
\newblock {\em Model-checking pushdown systems}.
\newblock PhD thesis, TUM, 2002.

\bibitem{Seth09}
A.~Seth.
\newblock Games on multi-stack pushdown systems.
\newblock In {\em LFCS'09}, volume 5407 of {\em LNCS}, pages 395--408.
  Springer, 2009.

\bibitem{LaTorre&Madhusudan&Parlato07}
S.~L. Torre, P.~Madhusudan, and G.~Parlato.
\newblock A robust class of context-sensitive languages.
\newblock In {\em LICS'07}, pages 161--170. IEEE, 2077.

\bibitem{LaTorre&Napoli11}
S.~L. Torre and M.~Napoli.
\newblock Reachability of multistack pushdown systems with scope-bounded
  matching relations.
\newblock In {\em CONCUR'11}, volume 6901 of {\em LNCS}, pages 203--218.
  Springer, 2011.

\bibitem{LaTorre&Napoli12}
S.~L. Torre and M.~Napoli.
\newblock A temporal logic for multi-threaded programs.
\newblock In {\em TCS 2012}, volume 7604 of {\em LNCS}, pages 225--239.
  Springer, 2012.

\bibitem{Vardi&Wolper94}
M.~Vardi and P.~Wolper.
\newblock Reasoning about infinite computations.
\newblock {\em I\&C}, 115:1--37, 1994.

\bibitem{Walukiewicz01}
I.~Walukiewicz.
\newblock Pushdown processes: games and model-checking.
\newblock {\em I\&C}, 164(2):234--263, 2001.

\end{thebibliography}
